\documentclass[12pt]{article}
\usepackage{graphicx}
\usepackage{amsmath}
\usepackage{amsthm}
\usepackage{amssymb}
\usepackage{mathtools}
\usepackage{dsfont}
\usepackage{geometry}
\usepackage{fontspec}
\usepackage{pgfplots}
\usepackage{xcolor}
\pgfplotsset{compat=1.18}
\usetikzlibrary{arrows.meta}
\usepackage[backend=biber,style=authoryear,maxnames=3,giveninits=false,uniquename=false]{biblatex}
\usepackage[onehalfspacing]{setspace}

\newtheorem{theorem}{Theorem}
\newtheorem{lemma}{Lemma}
\newtheorem{proposition}{Proposition}
\newtheorem{corollary}{Corollary}
\newtheorem{assumption}{Assumption}
\theoremstyle{definition}
\newtheorem{example}{Example}
\theoremstyle{plain}
\makeatletter
\newcommand{\indentafterthm}{\@afterindenttrue\@afterheading}
\AfterEndEnvironment{theorem}{\indentafterthm}
\AfterEndEnvironment{lemma}{\indentafterthm}
\AfterEndEnvironment{proposition}{\indentafterthm}
\AfterEndEnvironment{corollary}{\indentafterthm}
\makeatother

\pgfplotsset{
  theoryaxis/.style={
    axis lines=middle,
    axis line style={-Stealth, line width=0.8pt},
    tick style={black},
    tick align=outside,
    xlabel style={font=\small},
    ylabel style={font=\small},
    ticklabel style={font=\small,/pgf/number format/fixed},
    xmin=0, xmax=1, ymin=0, ymax=1,
    width=8cm, height=6cm,
    legend cell align=left
  }
}

\title{Prior-Free Delegation\footnote{We would especially like to thank Sarah Auster, Andreas Kleiner, Daniel Krähmer, and Benny Moldovanu for helpful discussions and comments. Terrence McGovern would like to thank the German Research Foundation (DFG) for support through CRC TR-224 EPoS (Project B02) and Dirk Bergemann for hosting him at Yale, where this work was partially completed. Wenjun Zheng gratefully acknowledges funding from the German Research Foundation (DFG) through Germany's Excellence Strategy (EXC 2126/1-390838866).}}
\author{
  Terrence McGovern\footnote{University of Bonn, Email: mcgovern@uni-bonn.de}
  \and
  Jan Benedikt Napp\footnote{University of Bonn, Email: jbnapp@uni-bonn.de}
  \and Wenjun Zheng\footnote{University of Bonn, Email: wenjun@uni-bonn.de}
}
\date{\today}
\begin{document}
\maketitle
\begin{abstract}
We consider a robust delegation problem in which the principal does not know the distribution from which the underlying state is drawn. The principal can choose a general randomized mechanism and maximizes her worst-case expected payoff over all state distributions. Our main result characterizes the robustly optimal mechanism. The mechanism has up to three regions: (i) accommodation, where the agent’s ideal action is taken, (ii) calibrated randomization, where each type receives a distinct lottery, and (iii) pooling, where all types receive the same lottery. We then extend our model to a multidimensional setting in which preferences are additively separable and satisfy a cross-dimensional symmetry condition, and show that delegating separately across dimensions is robustly optimal.
\end{abstract}
\section{Introduction}
\label{sec:introduction}
A principal may benefit from delegating a decision to a better-informed agent whose objectives differ from her own. A firm's owner may delegate investment decisions to a manager who is better informed about local demand. A policymaker may leave regulatory decisions to an agency with more detailed knowledge of the sector. In the canonical delegation model, the principal knows how the underlying state is distributed and allows the agent to choose from a subset of the available actions. In many such settings, however, the principal has little experience with the environment, such as when a firm enters a new market or a policymaker has to solve a novel regulatory problem. The principal may therefore face considerable uncertainty about the circumstances the agent will encounter and may be unable or unwilling to specify a prior over states. \par

We consider a delegation problem in which an agent privately observes a state and a principal controls an action. Both players have quadratic preferences around ideal actions that increase with the state. We assume that the agent's bias, defined as the difference between the agent's ideal action and the principal's, is monotone and that the principal knows the agent's preferences. We depart from the canonical delegation model in two ways. First, the principal knows neither the state nor its distribution. Second, we allow the principal to commit to a general randomized mechanism, which assigns a lottery over actions to each report. There are no monetary transfers. The principal evaluates a mechanism by its lowest expected payoff across all state distributions. We model this as a zero-sum game in which the principal chooses a mechanism and Nature chooses the distribution to minimize the principal's expected payoff. \par

With quadratic preferences, each player's expected utility at a given state depends only on the mean and variance of the implemented action. Holding the mean fixed, both players' utilities are decreasing in the variance. Yet the principal can use differences in variance across reports to screen the agent. By offering different combinations of mean and variance, she can distinguish types without assigning each type his ideal action. Variance plays a role similar to money burning: it allows the principal to adjust the utility associated with a mean action without making a monetary transfer. This gives her more flexibility in choosing the mean function, allowing her to respond to the agent's information while keeping the mean action closer to her ideal action. The added variance, however, reduces the principal's own payoff. The optimal mechanism determines where this greater control over the mean function is worth the cost of randomization. \par

Our first main result characterizes the robustly optimal mechanism under positive bias. The optimal mechanism features up to three regions: accommodation, calibrated randomization, and pooling. In the \emph{accommodation region} at low states, the principal assigns the agent his ideal action with zero variance, because disagreement is small enough for accommodation to preserve the principal's guarantee. In the \emph{calibrated-randomization region}, the mechanism assigns a distinct lottery to each report and uses variance to keep the mean function below the agent's ideal action while continuing to respond to his information. At high reports, the mechanism enters the \emph{pooling region}, where it assigns a common lottery because further screening is too costly. \par

We establish optimality by constructing a state distribution under which no incentive-compatible mechanism gives the principal a higher expected payoff than the guarantee attained by our mechanism. The mechanism and this distribution therefore form a saddle point. We extend the characterization to negative bias, where the agent prefers lower actions than the principal. The order of the regions is reversed: pooling occurs at low states, calibrated randomization at intermediate states, and accommodation at high states whenever it is present. The principal now uses randomization to keep the mean function above the agent's ideal action. \par

We then consider ideal-action functions that cross. When the agent's ideal action is more responsive to the state than the principal's, the bias changes from negative to positive. The agent overreacts to changes in the state from the principal's perspective: below the point of agreement he prefers a lower action than the principal, and above it he prefers a higher action. When his ideal action is less responsive, the bias changes from positive to negative. The agent is conservative in the sense that his preferred action remains closer to the action at which the players agree than the principal's does. \par

We obtain the robustly optimal mechanism by splitting the state space at the point of agreement and solving the delegation problem on each side. With increasing bias, both regional mechanisms accommodate the agent near the crossing and can be joined directly. With decreasing bias, their mean functions meet at the common ideal action, but their variances can differ. We combine the regional mechanisms by leaving their mean functions unchanged and adding the difference between the regional worst-case losses to the variance at every report in the region with the smaller loss. This preserves incentives within each region and makes the agent's truthful utility continuous at the crossing, so there is no profitable deviation to a report on the other side. In both cases, the overall guarantee is determined by the regional problem with the larger worst-case loss. \par

With a more responsive agent, the principal seeks to restrain excessive responses to states away from the point of agreement. Near the crossing, disagreement is small and the optimal mechanism accommodates the agent. Further away, the mechanism uses randomization to limit his response, with pooling at the extremes. The mechanism therefore assigns higher variance to more extreme reports, while the common ideal action at the crossing is implemented without randomization. \par

With a conservative agent, the principal instead wants actions to respond more strongly to the state than the agent. This changes the role of variance in the mechanism. Offering the common ideal action with zero variance would make central reports too attractive to types whom the principal wants to induce to choose more extreme mean actions. The optimal mechanism assigns higher variance to reports near the crossing and lower variance to reports further away, making these more responsive allocations attractive to the agent. The principal therefore uses randomization even at the state where their ideal actions agree. \par

In practice, a principal often delegates several decisions to the same agent. In addition to not knowing the marginal state distribution for each decision, the principal may not know how the states are correlated across decisions. We extend our analysis to this setting, where the principal maximizes the worst-case expected payoff over all joint distributions of the states. This allows for uncertainty about both the marginal distributions and the dependence between states. We characterize a robustly optimal mechanism when preferences are additively separable across dimensions and each player’s ideal-action functions in each dimension are transformations of the player’s reference ideal-action function.  This includes the case in which each player’s ideal-action function is identical across dimensions. We show that separate delegation is robustly optimal under these conditions. This simple rule applies the optimal one-dimensional mechanism separately in each dimension, so the lottery assigned in a given dimension depends only on the state reported for that dimension.  \par\medskip

\textbf{Literature Review.}
Our paper contributes to the literature on delegation initiated by \textcite{holmstrom1984delegation}. The literature has mainly considered a principal who chooses a delegation set, from which an informed agent selects his preferred action (\textcite{melumad1991communication}, \textcite{martimort2006continuity}, \textcite{alonso2008optimal}, \textcite{amador2013delegation}, \textcite{amador2025generalized}). We allow the principal to commit to a mechanism that randomizes over actions for each report, as in \textcite{kovavc2009stochastic}, \textcite{kleiner2021extreme}, and \textcite{kolotilin2025persuasion}.\footnote{See also related work on money burning and nonmonetary incentives (\textcite{ambrus2017nonmonetary}, \textcite{amador2020money}).} \par

Our paper is also related to the growing literature on robustness in mechanism design (\textcite{bergemann2005robust}, \textcite{bergemann2008pricing}, \textcite{bergemann2011pricing}, \textcite{carroll2015linear}).\footnote{See \textcite{carroll2019survey} for a survey.} Within delegation, \textcite{frankel2014aligned} studies a principal who knows the state distribution but is uncertain about the agent's preferences. \textcite{alonso2026robust} consider a multidimensional problem in which the principal knows the agent's ideal action in each state but has limited knowledge of the agent's utility function. Under their restrictions on preference uncertainty, they establish the max-min optimality of convex delegation. This result continues to hold when ambiguity also covers all state distributions. We instead take preferences as known and optimize over arbitrary incentive-compatible lotteries over actions. \par

A related paper on robustness with respect to the state distribution is \textcite{jia2026random}. The paper studies quadratic preferences with constant bias and characterizes an optimal randomized mechanism and a worst-case distribution when the principal knows the mean of the state. Our analysis characterizes the robustly optimal mechanism for general monotone bias, including preferences that cross. \textcite{carrasco2013robust} study delegation under ambiguity about the state distribution and show that full delegation is optimal in their setting with an unbounded state space. \textcite{guo2023regret} study a regret-minimization problem in which the principal does not know which projects are available to the agent. \par

Finally, our paper relates to the literature on delegating multiple decisions (\textcite{koessler2012multidimensional}, 
\textcite{frankel2014aligned}, \textcite{frankel2016multiple}, \textcite{kleiner2025delegation}). \textcite{kleiner2025delegation} uses a convex representation of the agent's indirect utility to characterize incentive compatibility and optimality in a multidimensional model with randomized mechanisms.

A central question in this literature is whether the principal benefits from linking decisions (\textcite{jackson2007linking}). \textcite{frankel2016multiple} shows that linking can benefit the principal even with finitely many decisions. With quadratic preferences and independently normally distributed states, the optimal delegation set constrains a weighted sum of actions, allowing greater discretion in one decision to be offset by restraint in another. Whether linking helps depends on what the principal is uncertain about. The max-min optimal mechanisms in \textcite{frankel2014aligned}, such as rankings, budgets, and sequential quotas, all link decisions together. Under ambiguity about the joint state distribution, we find that the principal gains nothing from linking. \textcite{ball2026quota} study quota mechanisms with finitely many decisions and states drawn independently across decisions; they bound the decision errors caused by sampling variation and misspecified quotas and show that their guarantee cannot be improved by other linking mechanisms. \textcite{carroll2017separation} provides a related robust separation result in screening with monetary transfers: when the principal knows the marginal type distributions but not their dependence, separate screening is optimal. We take the joint distribution to be entirely unknown and ask when separate delegation is optimal.

The paper is organized as follows. Section~\ref{sec:model} introduces the model. Section~\ref{sec:main-results} characterizes the optimal mechanism and worst-case distribution. Section~\ref{sec:multiple-dimensions} extends the analysis to multiple dimensions. Appendix~\ref{app:main-results} contains the main proofs. Appendix~\ref{app:negative-bias} contains the reflection argument for negative bias. \par

\section{Model}
\label{sec:model}

There is an agent (he) who observes the underlying state of the world $\theta\in\Theta=[0,1]$, with cumulative distribution function (CDF) $F$. There is a principal (she) who controls an action $a\in\mathcal{A}\subset\mathbb{R}$. The principal neither knows $\theta$ nor the distribution $F$ from which $\theta$ is drawn. Let $\mathcal{F}(\Theta)$ denote the set of CDFs of Borel probability distributions on $\Theta$. Write $\operatorname{supp}(F)$ for the support of the Lebesgue-Stieltjes measure associated with $F$.

Denote the principal's and agent's ideal-action functions by
\[
    y_P,y_A:[0,1]\rightarrow\mathcal{A},
\]
respectively. These represent the ideal action each player would like to take in a given state. We assume that both ideal-action functions are continuously differentiable and strictly increasing. The action space $\mathcal{A}$ is a compact interval of the real line and is assumed to be sufficiently large to implement the action lotteries constructed below. \par
Preferences are given by quadratic loss. The principal's and agent's realized (ex post) utilities are, respectively,
\begin{align*}
    u_P(a,\theta)
    &=-\bigl(a-y_P(\theta)\bigr)^2,\\
    u_A(a,\theta)
    &=-\bigl(a-y_A(\theta)\bigr)^2.
\end{align*}

We denote the bias by
\[
    b(t):=y_A(t)-y_P(t).
\]
We assume that the bias is either weakly increasing or weakly decreasing and satisfies a strict single-crossing property: there is at most one state $\theta\in[0,1]$ such that $b(\theta)=0$.

The principal can choose any randomized direct mechanism, i.e., a Borel
probability kernel
\begin{equation*}
    m:\Theta\rightarrow\Delta(\mathcal{A}).
\end{equation*}
A mechanism maps a reported state $\hat\theta$ to a probability distribution $m(\hat\theta)$ over actions.\footnote{We let
$\Delta(\mathcal{A})$ denote the set of Borel probability distributions on $\mathcal{A}$ and write $dm(a\mid\hat\theta)$ for integration with respect to the action distribution $m(\hat\theta)$.} We denote the set of all direct mechanisms by $\tilde{\mathcal{M}}$. \par
The principal solves a robust maximization problem. Specifically, for any mechanism $m$, the principal evaluates its payoff by taking the infimum over all $F\in\mathcal{F}(\Theta)$. We represent this worst-case evaluation as a zero-sum game in which Nature chooses a state distribution
to minimize the principal's expected payoff.

Because $\mathcal{A}$ is compact, every outcome lottery has finite first and second moments. We let
\begin{align}
    \mu_m(\hat\theta)
    &:=
    \int_{\mathcal{A}}a\,dm(a\mid\hat\theta),
    \label{eq:mechanism-mean}\\
    \tau_m(\hat\theta)
    &:=
    \int_{\mathcal{A}}
    \bigl(a-\mu_m(\hat\theta)\bigr)^2\,dm(a\mid\hat\theta)
    \label{eq:mechanism-variance}
\end{align}
denote the mean action and variance of the implemented action following the report $\hat\theta$, respectively. When the mechanism is clear from the context, we write $\mu$ for $\mu_m$ and $\tau$ for $\tau_m$.

For each player $j\in\{P,A\}$, let
\begin{equation*}
    U_{j,m}(\theta,\hat\theta):=\int_{\mathcal{A}}u_j(a,\theta)\,dm(a\mid\hat\theta)
\end{equation*}
denote expected utility when the agent reports $\hat\theta$ at state $\theta$. The mechanism is incentive compatible if the agent finds it optimal to report the true
state, or
\begin{equation}
    U_{A,m}(\theta,\theta)
    \geq U_{A,m}(\theta, \hat\theta)
    \qquad
    \text{for every }\theta,\hat\theta\in\Theta.
    \tag{IC}
    \label{eq:IC}
\end{equation}
We denote the set of incentive-compatible direct mechanisms by $\mathcal{M}$. When the mechanism is incentive compatible, we abuse notation slightly and write $U_{j,m}(\theta)$ for $U_{j,m}(\theta,\theta)$, for $j\in\{P,A\}$. We assume that whenever truthful reporting is among the agent's optimal reports, the agent reports truthfully. 

Given an incentive-compatible mechanism, the principal's ex ante utility for a specific CDF $F$ is
\begin{align*}
    V(m,F)
    &:=\int_{\Theta}U_{P,m}(\theta)\,dF(\theta)\\
    &=\int_{\Theta}\int_{\mathcal{A}}
    u_P(a,\theta)\,dm(a\mid\theta)\,dF(\theta).
\end{align*}
The principal's objective is then
\[
    \sup_{m\in\mathcal{M}}
    \inf_{F\in\mathcal{F}(\Theta)}
    V(m,F).
\]
First, the principal and Nature simultaneously commit to a
mechanism $m$ and a CDF $F$, respectively. Then a state $\theta$ is realized according to $F$, which is observed by the agent, who then sends a report to the mechanism. Finally, an allocation is implemented by the mechanism and payoffs are realized.

\section{Main Results}
\label{sec:main-results}

We first characterize incentive compatibility and state the saddle-point
criterion. We then present the saddle point under positive bias, provide
some intuition, and illustrate the mechanism with examples.
In Section~\ref{subsec:negative-crossing-bias-saddle}, we extend the result
to a sign-changing bias. All proofs are relegated to the
appendix.

\subsection{Incentive Compatibility and Saddle Points}
\label{subsec:ic-saddle}

With quadratic preferences, we can decompose expected utility into the
variance and the squared distance between the mean action and the
player's ideal action:
\begin{equation}
    U_j(\theta)
    =-\tau(\theta)
    -\bigl(\mu(\theta)-y_j(\theta)\bigr)^2,
    \qquad j\in\{P,A\}.
    \label{util_decomp}
\end{equation}
The mean and variance therefore determine each player's expected
utility under truthful reporting.

We now characterize the mean and variance functions of an
incentive-compatible mechanism.
The following lemma states the standard characterization of incentive
compatibility in this setting, due to
\textcite{kovavc2009stochastic}.
Our proof adapts their argument to the preferences in our model.

\begin{lemma}
\label{lem:ic-characterization}
Let $\mu$ and $\tau$ denote the mean and variance of a mechanism
$m\in\tilde{\mathcal{M}}$. Then $m$ is incentive compatible if and
only if:
\begin{align*}
    (\mathrm{IC}_1)\quad
    &\mu(\theta)\text{ is nondecreasing on }\Theta;\\[1mm]
    (\mathrm{IC}_2)\quad
    &U_A(\theta_2,\theta_2)-U_A(\theta_1,\theta_1)
    =2\int_{\theta_1}^{\theta_2}
    [\mu(t)-y_A(t)]\,y_A'(t)\,dt
    \\
    &\hspace{1cm}\text{for every }\theta_1,\theta_2\in\Theta;\\[1mm]
    (\mathrm{VAR})\quad
    &\tau(\theta)\geq 0
    \quad\text{for every }\theta\in\Theta.
\end{align*}
\end{lemma}

The lemma gives two restrictions on how the mechanism can respond to
the agent's report. The mean action must be nondecreasing, and the
agent's truthful utility must satisfy the envelope condition
$(\mathrm{IC}_2)$. Given the mean function and the agent's utility at
one state, this condition determines his truthful utility at every
other state. The utility decomposition in \eqref{util_decomp} then
determines the variance. 

We solve the principal's robust problem by finding a mechanism and a
CDF that are best responses to each other. Such a pair $(m^*,F^*)$
is a saddle point if
\begin{equation*}
    V(m^*,F)\geq V(m^*,F^*)\geq V(m,F^*)
    \label{eq:saddle}
\end{equation*}
for every $m\in\mathcal{M}$ and $F\in\mathcal{F}(\Theta)$.
The first inequality means that $m^*$ guarantees the principal at least
$V(m^*,F^*)$, whatever distribution Nature chooses. The second means
that, under $F^*$, no incentive-compatible mechanism gives her higher utility.
Thus, Nature can prevent the principal from improving on the guarantee
provided by $m^*$. This establishes that $m^*$ solves the robust
problem and that its value is $V(m^*,F^*)$.

\subsection{The Robust Mechanism  under Positive Bias}
\label{subsec:positive-bias-saddle}
We now assume $b(\theta)\geq0$ for every $\theta\in\Theta$ and state the saddle point under this restriction. We say that the principal and agent have \textit{overlapping ideal-action ranges} if $y_P(\Theta)\cap y_A(\Theta)$ contains a nondegenerate interval. When the bias is positive, this is equivalent to $y_A(0)< y_P(1)$. We first consider the case in which players have \textit{non-overlapping ideal-action ranges}. 

\paragraph{Non-overlapping ideal-action ranges} When $b(\theta)>0$ the preferences satisfy \textit{non-over-lapping ideal-action ranges} when $y_A(0)\geq y_P(1)$. We first characterize the robustly optimal mechanism in this setting:
\begin{proposition}
\label{prop:complete-pooling}
Suppose the preferences satisfy \textit{non-overlapping ideal-action ranges}. The mechanism $m^*$
which assigns the deterministic action
\begin{equation*}
    \mu(\theta)=\frac{y_P(0)+y_P(1)}{2}
    \label{eq:complete-pooling-mechanism}
\end{equation*}
after every report is robustly optimal.
\end{proposition}
The proposed mechanism assigns the same deterministic action after every report. This action lies halfway between the principal's lowest and highest ideal actions, giving her the same utility at the two extreme states and higher utility at every interior state. The worst-case distribution can then only be supported on the two extreme states, and must make it a best response for the principal to implement the average of her ideal actions at the two endpoint states. The following distribution therefore constitutes a saddle point with $m^*$:
\[
F^*(\theta)=
    \begin{cases}
        0,&\theta<0,\\
        \tfrac12,&0\leq\theta<1,\\
        1,&\theta\geq1.
    \end{cases}
    \label{eq:complete-pooling-cdf}
\]
The principal cannot gain by offering different lotteries at the two
endpoint reports. To see this, take any incentive-compatible mechanism
$m$ and first suppose $y_A(0)=y_P(1)$. The agent at $0$ and the
principal at $1$ then have identical preferences over lotteries, so
incentive compatibility gives
\begin{equation*}
    U_{P,m}(1,0)=U_{A,m}(0,0)
    \geq U_{A,m}(0,1)=U_{P,m}(1,1).
    \label{eq:pooling-lottery-comparison}
\end{equation*}
In any incentive-compatible mechanism, the principal at state $1$ weakly prefers the lottery assigned after report $0$ to the lottery assigned after report $1$. The principal therefore weakly improves her payoff by also assigning the zero report lottery to the report of one. The incentive problem is even more severe when $y_A(0)>y_P(1)$ since any increase in the mean between the two lotteries improves the agent's utility at state zero by more than the principal's utility at state one. 

\paragraph{Overlapping ideal-action ranges} We now consider preferences with overlapping ideal-action ranges. We first describe the optimal mechanism and state our main result. We then provide some intuition and examples. \par
The mechanism features up to three regions. We specify these regions using two cutoffs which we denote by $\underline{\theta}$ and $\bar{\theta}$. The upper cutoff is defined as
\[
\bar{\theta}:=y_A^{-1}\bigl(y_P(1)\bigr).
\]
This is the state at which the agent's ideal action equals the principal's ideal action in the highest state. Given that $y_A(\theta)$ is strictly increasing, it exists uniquely and lies in $(0,1]$. \par
To describe the lower cutoff and the robustly optimal mechanism, we first introduce the following function:
\[
d(\theta)
    :=\frac{b(\bar{\theta})}{2}
    \exp\!\left(-\int_\theta^{\bar{\theta}}
    \frac{y_P'(t)}{b(t)}\,dt\right)+\int_\theta^{\bar{\theta}}y_P'(t)
    \exp\!\left(-\int_\theta^t
    \frac{y_P'(s)}{b(s)}\,ds\right)\,dt.
    \label{eq:positive-distance}
\]
The function $d(\theta)$ describes the difference between the average action implemented under report $\theta$ and the principal's ideal action for the intermediate region as we will formally define below.\footnote{If $b(1)=0$, the upper cutoff is $\bar{\theta}=1$ and the integrals with upper endpoint $1$ are understood as improper integrals. We evaluate them by their one-sided limits and use $\exp(-\infty)=0$. This convention also applies to the CDF below, while the distance at the endpoint is $d(1)=0$.}
We can then define the lower cut-off. Define
\[
\underline{\theta}:=\inf\{\theta\in[0,\bar{\theta}]:
        b(\theta)>0,\ d(\theta)\leq b(\theta)\}.
\]
The lower cutoff is the first state with positive bias at which the candidate mean $y_P(\theta)+d(\theta)$ is weakly below the agent's ideal action. We now define the robustly optimal mechanism $m^* = (\mu^*,\tau^*)$:
\[
\mu^*(\theta)=
    \begin{cases}
        y_A(\theta),&0\leq\theta<\underline{\theta},\\[1mm]
        y_P(\theta)+d(\theta),
            &\underline{\theta}\leq\theta\leq\bar{\theta},\\[1mm]
        \dfrac{y_P(\bar{\theta})+y_P(1)}{2},
            &\bar{\theta}<\theta\leq1,
    \end{cases}
    \text{ and }\tau^*(\theta) = 
    \begin{dcases}
0,&0\leq\theta<\underline{\theta},\\[1mm]
d(\underline{\theta})^2-d(\theta)^2,&\underline{\theta}\leq\theta\leq\bar{\theta}\\[1mm]
d(\underline{\theta})^2-\frac{b(\bar{\theta})^2}{4},&\bar{\theta}<\theta\leq 1
    \end{dcases}
\]
\begin{theorem}
\label{thm:positive-bias-saddle}
Suppose the preferences satisfy \textit{overlapping ideal-action ranges}. The mechanism $m^*$ is robustly optimal. The value
of the robust problem is
\[
V(m^*,F^*)=-d(\underline{\theta})^2.
\]
\end{theorem}
The proof of Theorem~\ref{thm:positive-bias-saddle} proceeds in three steps. First, we show that $m^*$ is incentive compatible. We then show that $m^*$ guarantees expected utility at least $V(m^*,F^*)$ for all possible distribution functions. We then construct a distribution $F^*$, which guarantees the principal's expected utility is no higher than $V(m^*,F^*)$. Thus, $(m^*,F^*)$ forms a saddle point.
\par
\begin{figure}[!t]
    \centering
    \includegraphics[width=0.9\linewidth]{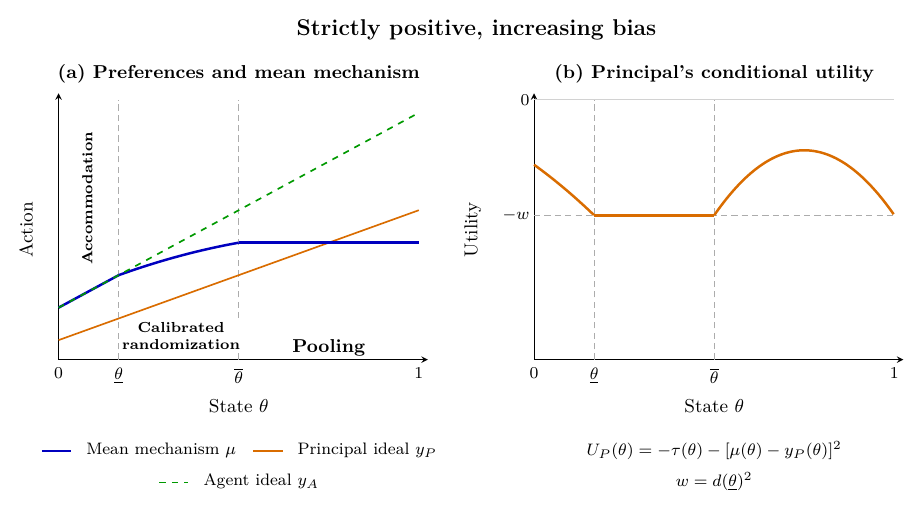}
    \caption{The robustly optimal mechanism with increasing bias: (a) the mean action; (b) the principal's utility at each state.}
    \label{fig:increasing-bias-utility}
\end{figure}
We now explain the robustly optimal mechanism. Figure~\ref{fig:increasing-bias-utility}(a) plots the robustly optimal mechanism for increasing bias. Figure~\ref{fig:increasing-bias-utility}(b) shows the principal's utility at each state given the robustly optimal mechanism. We describe each of the three regions featured in the mechanism in turn. The \textit{accommodation region} $[0,\underline{\theta})$ is nonempty when the bias is increasing and disagreement at low states is sufficiently small that implementing the agent's ideal action gives the principal more than her guarantee of $-d(\underline{\theta})^2$. In this region, the mechanism assigns the deterministic action $y_A(\theta)$, giving the principal utility $-b(\theta)^2$. As the state approaches $\underline{\theta}$, the payoff from accommodation falls to the value of the robust problem. The principal can therefore accommodate the agent throughout this region while preserving her worst-case guarantee. With decreasing bias, the accommodation region is empty. \par
For $\theta\in[\underline{\theta},\bar{\theta}]$, the mechanism uses \textit{calibrated randomization}. The mean is $\mu^*(\theta)=y_P(\theta)+d(\theta)$, where $d(\theta)$ measures the distance between the mean action and the principal's ideal action. We can understand this distance by solving backward from the pooling region. At the upper cutoff, the pooling lottery fixes the mean halfway between $y_P(\bar{\theta})$ and $y_P(1)$. Since $y_A(\bar{\theta})=y_P(1)$, this gives $d(\bar{\theta})=b(\bar{\theta})/2$. Starting from this boundary, the envelope condition and the requirement that the principal's utility remain constant jointly determine $d(\theta)$. Moving backward toward lower states, the variance weakly decreases while the distance from the principal's ideal action weakly increases. The reduction in variance exactly offsets the increase in squared distance, keeping the principal's utility at $-d(\underline{\theta})^2$ throughout the region. When accommodation is present, this backward construction reaches the agent's ideal action at the lower cutoff, where $d(\underline{\theta})=b(\underline{\theta})$ and the variance is zero. The mechanism then joins continuously with deterministic accommodation below the cutoff. Otherwise, calibrated randomization extends to state $0$. \par
We say that $(\bar{\theta},1]$ constitutes the pooling region, as all types receive the same lottery. The intuition follows from Proposition~\ref{prop:complete-pooling}. The agent's ideal action at $\bar{\theta}$ is equal to the principal's optimal action at one. They have the same preferences over lotteries, so under any IC mechanism the principal at state $1$ weakly prefers the lottery assigned after report $\bar{\theta}$ to the lottery assigned after report $1$. She therefore assigns a common lottery to all types in this interval which has a mean halfway between the two optimal actions at the end-points of the interval. \par 
We can see all three intervals in Figure~\ref{fig:increasing-bias-utility}. At low states the bias is relatively small and so the mechanism chooses the optimal action for the agent. The principal's payoff is above the guaranteed level in this region, so decreasing the mean action, which would tighten incentive constraints, is not beneficial. In the calibrated randomization region the payoff to the principal is constant at $-d(\underline{\theta})^2$. For a nonempty accommodating region, the variance is zero at $\underline{\theta}$ and the mean is equal to the agent's ideal action. At $\bar{\theta}$ the agent's ideal action is equal to the principal's ideal action at the highest state. The mean of the mechanism therefore is constant across $[\bar{\theta},1]$ and the principal's utility first increases and then decreases again, such that it equals $-d(\underline{\theta})^2$ at $\theta=1$.\par
\begin{figure}[!t]
    \centering
    \includegraphics[width=0.9\linewidth]{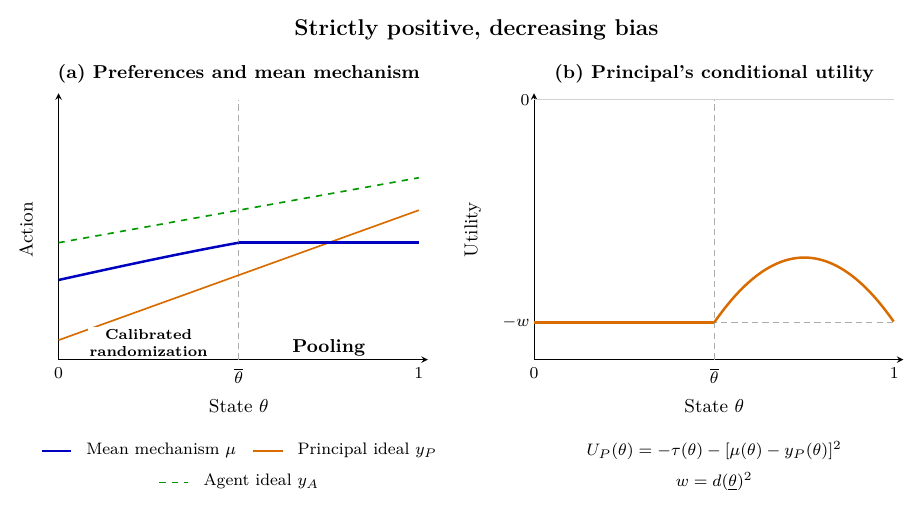}
    \caption{The robustly optimal mechanism with decreasing bias: (a) the mean action; (b) the principal's utility at each state.}
    \label{fig:decreasing-bias-utility}
\end{figure}
Figure~\ref{fig:decreasing-bias-utility}(a) shows a case where the preferences are such that the bias is decreasing. Figure~\ref{fig:decreasing-bias-utility}(b) shows the corresponding utility of the principal in each state. With a decreasing bias, the difference in ideal actions is larger at lower states. Implementing the agent's ideal action at these states is therefore not beneficial, and the mechanism only features the calibrated randomization and pooling regions. With constant bias $b(\theta)=\lambda\in(0,1)$ for all $\theta\in[0,1]$, the robustly optimal mechanism features calibrated randomization and pooling as we will show in Example \ref{ex:constant-bias}. \par 
By construction, the mechanism $m^*$ gives the principal the same utility $-d(\underline{\theta})^2$ at every state in $[ \underline{\theta},\bar{\theta}]\cup\{1\}$, and weakly higher utility at all other states. Therefore, any distribution supported on this set is a best response by Nature. The substantive problem is thus to identify a distribution $F^*$ under which $m^*$ is a best response; this is fulfilled by
\[
F^*(\theta)=
    \begin{cases}
        0,&\theta<\underline{\theta},\\[1mm]
        1-\dfrac{d(\underline{\theta})}{b(\underline{\theta})}
        \exp\!\left(-\displaystyle\int_{\underline{\theta}}^{\min\{\theta,\bar{\theta}\}}
        \frac{y_A'(t)}{b(t)}\,dt\right),
            &\underline{\theta}\leq\theta<1,\\[1mm]
        1,&\theta\geq1,
    \end{cases}
\]
The relevant property of $F^*$ is that throughout the calibrated randomization region, 
\[y_A(\theta)= \mathbb{E}_{F^*}[y_P(t)|t>\theta] \quad \forall \theta \in [\underline{\theta},\bar{\theta}]\]
So conditional on the state being higher than $\theta$, the principal's expected ideal action under $F^*$ coincides with the agent's ideal action at $\theta$.\footnote{If $\bar{\theta}=1$, the conditional expectation is defined only for $\theta<1$, since the event $t>1$ has probability zero. At $\theta=1$, the displayed identity is understood in the limiting sense: as $\theta\uparrow1$, the conditional expectation converges to $y_P(1)=y_A(1)$.} The property relates to the forward bias definition by \textcite{alonso2008optimal}; let $S_F$ be the forward bias of $F$ given $y_A$ and $y_P$; then the above equality is equivalent to $S_{F^*}(\theta)=0$ on this range. Importantly, this property makes the principal indifferent to marginal changes in the mean allocation $\mu(\theta)$ over $[\underline{\theta},\bar{\theta}]$ and pins down the distribution's hazard rate to equal $y_A^\prime(\theta)/b(\theta)$ in the interior of this range.

The following example gives the mechanism and distribution under constant bias.
\begin{example}
\label{ex:constant-bias}
Suppose $y_P(\theta)=\theta$ and $y_A(\theta)=\theta+\lambda$,
where $\lambda>0$. If $0<\lambda<1$, then
$\underline{\theta}=0$, $\bar{\theta}=1-\lambda$.
The optimal mechanism $m^*$ has the associated mean function:
\begin{align*}
    \mu^*(\theta)
    &=\begin{cases}
        \theta+\lambda\left(
        1-\frac12\exp\!\left(\frac{\theta-(1-\lambda)}{\lambda}\right)
    \right),&0\leq\theta\leq1-\lambda,\\[1mm]
        1-\lambda/2,&1-\lambda<\theta\leq1.
    \end{cases}
    \label{eq:constant-bias-mean}
\end{align*}
Define
\begin{equation*}
    F^*(\theta)=
    \begin{cases}
        0,&\theta<0,\\[1mm]
        1-\left(1-\dfrac12 e^{-(1-\lambda)/\lambda}\right)
            e^{-\min\{\theta,1-\lambda\}/\lambda},&0\leq\theta<1,\\[1mm]
        1,&\theta\geq1.
    \end{cases}
    \label{eq:constant-bias-cdf}
\end{equation*}
\end{example}
\subsection{The Robust Mechanism under Crossing Preferences}
\label{subsec:negative-crossing-bias-saddle}
The results from Theorem \ref{thm:positive-bias-saddle} can be applied to the setting with a negative bias by reflecting the state and action space. Reflection reverses the order of the three regions. Pooling occurs at low reports, calibrated randomization at intermediate reports, and accommodation at high reports. In the calibrated-randomization region, the mean lies above the agent's ideal action and the variance falls as the report increases. We provide more details in the appendix. \par

We now move on to the case where the principal's and agent's ideal-action functions cross. The robustly optimal mechanism can be found by combining the mechanisms from the two individual regions before and after the state at which they cross. Let $\theta_\times\in(0,1)$ denote the state at which they cross, so $y_A(\theta_\times)=y_P(\theta_\times)$. At this state the players agree on the optimal action, while their disagreement has opposite directions on the two sides. Consider the two problems obtained by restricting states and reports to $\Theta_1=[0,\theta_\times]$ and $\Theta_2=[\theta_\times,1]$.
Each is a special case of the preceding results: the bias has one sign, and the ideal actions agree at the right endpoint of $\Theta_1$ or the left endpoint of $\Theta_2$. With an increasing bias, the left problem has a negative bias and the right problem has a positive bias. The positive-bias result and its reflection therefore give a saddle point for each restricted problem. \par
Let $m_i^*$ denote the robustly optimal mechanism in region $i$, with mean $\mu_i^*$ and variance $\tau_i^*$. Write $w_i$ for its worst loss and $\bar{w}=\max\{w_1,w_2\}$ for the larger of the two. By Theorem~\ref{thm:positive-bias-saddle}, $w_i=d_i(\underline{\theta}_i)^2$, where $d_i$ and $\underline{\theta}_i$ are the distance function and lower cutoff used in the regional construction. For a negative-bias region, these refer to the reflected problem.\footnote{We treat $F_i^*$ as a CDF on $\Theta$ with all probability on $\Theta_i$: it is zero below that region and one at and above its upper endpoint.} \par
\begin{theorem}
\label{thm:crossing-bias-saddle}
The mechanism $m^*$ is robustly optimal and is given by the following mean and variance functions:
If the bias is increasing:
\begin{equation*}
    \bigl(\mu^*(\theta),\tau^*(\theta)\bigr)=
    \begin{cases}
        \bigl(\mu_1^*(\theta),\tau_1^*(\theta)\bigr), \;
            & 0\leq\theta\leq\theta_\times,\\[1mm]
        \bigl(\mu_2^*(\theta),\tau_2^*(\theta)\bigr), \;
            & \theta_\times<\theta\leq1,
    \end{cases}
    \label{eq:crossing-mechanism_incr}
\end{equation*}
and if the bias is decreasing:
\begin{equation*}
    \bigl(\mu^*(\theta),\tau^*(\theta)\bigr)=
    \begin{cases}
        \bigl(\mu_1^*(\theta),\tau_1^*(\theta)+\bar{w}-w_1\bigr),
            & 0\leq\theta\leq\theta_\times,\\[1mm]
        \bigl(\mu_2^*(\theta),\tau_2^*(\theta)+\bar{w}-w_2\bigr),
           & \theta_\times<\theta\leq1,
    \end{cases}
    \label{eq:crossing-mechanism_decr}
\end{equation*}
\end{theorem}
Theorem~\ref{thm:crossing-bias-saddle} shows how the two regional saddle points combine to solve the original problem. The combined mean function $\mu^*$ follows the regional mean on each side of $\theta_\times$. With an increasing bias, the regional lotteries can be used without adjustment. With a decreasing bias, the principal must add variance in the region with the smaller worst loss to preserve incentive compatibility at the crossing. Against this $m^*$, Nature chooses the CDF from the region with the larger worst loss. If the losses are equal, any mixture of the two regional CDFs can be used. More precisely, for any $\alpha\in[0,1]$: 
\[ F^*(\theta):= 
\begin{cases}
    F_1^*(\theta), & \text{if } w_1> w_2 \\  
    \alpha F_1^*(\theta) +(1-\alpha)F_2^*(\theta), \, & \text{if } w_1=w_2 \\
    F_2^*(\theta), & \text{if } w_1< w_2
\end{cases} 
\]
If the agent's preferences are more responsive than the principal's, the bias is increasing and the two regional mechanisms accommodate the
agent near $\theta_\times$ and can be joined without any adjustment.
The disagreement is small on both sides of the crossing, so assigning
the agent his ideal action keeps the principal's loss below the
respective regional worst loss. At the crossing, both mechanisms
therefore assign the common ideal action with zero variance. The agent's
truthful utility agrees at this point, and the regional means fit into
a nondecreasing global mean. Lemma~\ref{lem:ic-characterization} then
also rules out profitable reports across the crossing. \par

Figure~\ref{fig:combined-increasing-bias-asymmetric} shows the combined mechanism with an increasing bias. The principal assigns the agent his ideal action for states close to $\theta_\times$. On either side of this interval, the mechanism uses calibrated randomization, with pooling at the lowest and highest reports. The mean lies above the agent's ideal action to the left of the accommodation region and below it to the right. In this example, the principal's worst loss is larger on the right of $\theta_\times$, so Nature chooses $F_2^*$ and assigns no probability to the left-hand region.

\begin{figure}[!t]
    \centering
    \includegraphics[width=0.8\linewidth]{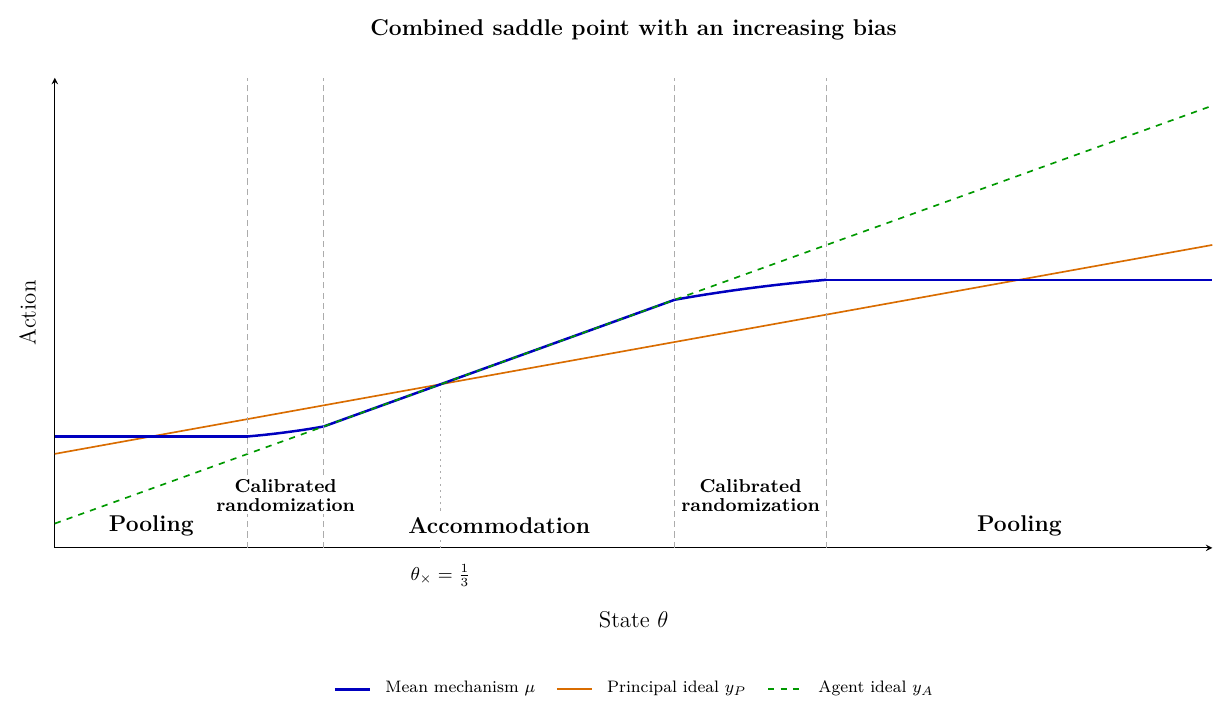}
    \caption{The combined saddle-point mean with an increasing bias,
    for $y_P(\theta)=\theta$ and $y_A(\theta)=2\theta-1/3$.
    The ideal actions cross at $\theta_\times=1/3$.}
    \label{fig:combined-increasing-bias-asymmetric}
\end{figure}

Figure~\ref{fig:combined-decreasing-bias-asymmetric} shows the case of a decreasing bias implying that the agent's preferences are less responsive than the principal's. Joining the unadjusted regional mechanisms need not preserve incentive compatibility. Their means still meet at the common ideal action, but their variances at the crossing are $w_1$ and $w_2$. If these differ, the agent's truthful utility approaches different values from the two sides. A type close to the crossing on the side with more variance could then obtain a similar mean with less risk by reporting a state on the other side. The principal must therefore adjust the lotteries to prevent these deviations.
\begin{figure}[!t]
    \centering
    \includegraphics[width=0.8\linewidth]{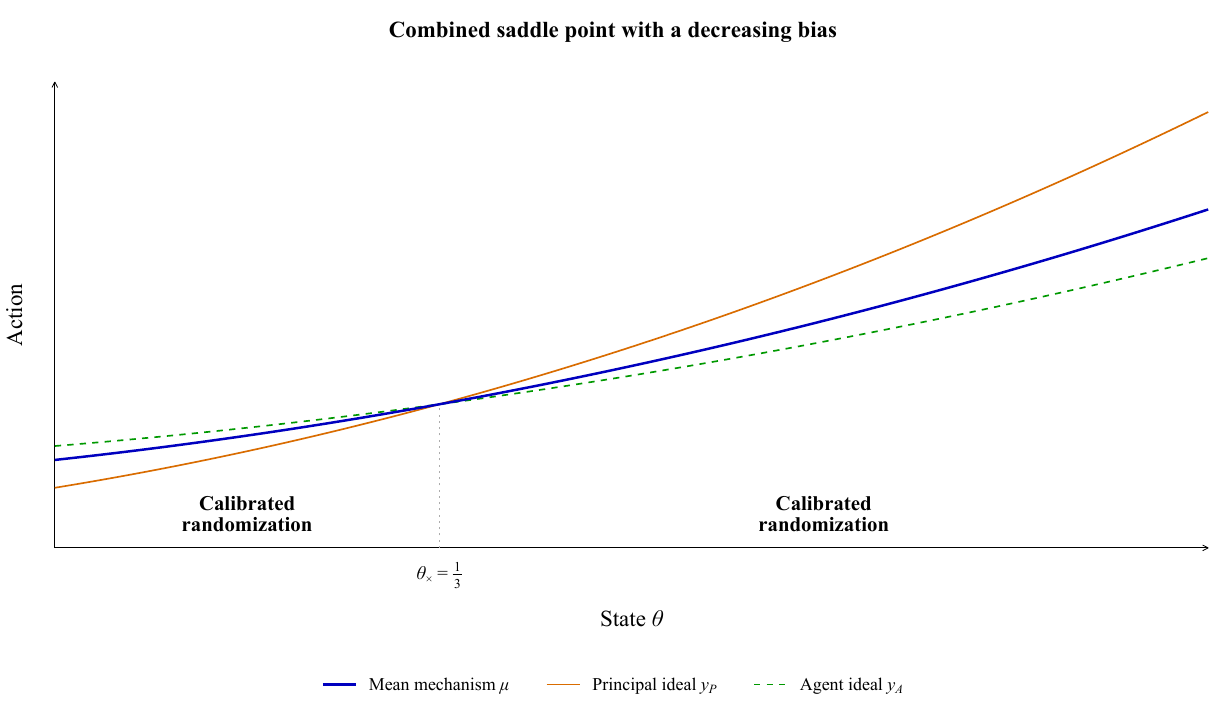}
    \caption{The combined saddle-point mean with a decreasing bias, for $y_P(\theta)=(\theta+\theta^2)/2$ and $y_A(\theta)=(\theta+\theta^2)/4+1/9$. The ideal actions cross at $\theta_\times=1/3$.}
    \label{fig:combined-decreasing-bias-asymmetric}
\end{figure}
We resolve this problem by increasing every variance in region $i$ by $\bar{w}-w_i$, while leaving the mean unchanged. This subtracts the same amount from the agent's utility at every report in that region, so his preferences between those reports do not change. At the crossing, both adjusted variances equal $\bar{w}$, and truthful utility agrees on the two sides. Together with the nondecreasing mean, the regional envelope conditions then give global incentive compatibility. The principal consequently gives up some payoff in the region with the smaller worst loss to preserve incentives between the two regions.

Nature chooses the CDF from the region with the lower saddle-point value. In this region, the mechanism is unchanged, so the principal's payoff is $-\bar{w}$. The principal cannot improve on this: any mechanism that is incentive compatible on the full state space is also feasible for the restricted problem, where $-\bar{w}$ is the optimal value. Nature cannot lower the payoff either, since the combined mechanism gives the principal at least $-\bar{w}$ in every state.
\section{Optimality of Separate Delegation}
\label{sec:multiple-dimensions}

In practice, a principal often delegates several decisions to the same agent. A firm's owner, for example, may delegate investment decisions across several projects to a manager who is better informed about their returns. The principal may in this case benefit from being able to link the actions taken across different decisions. We extend our model to the multidimensional environment. We define the state space as the product of the $n$ one-dimensional state spaces, and similarly for the action space:
\[
\Theta:=\prod_{i=1}^n \Theta_i=[0,1]^n,
\qquad
\mathcal{A}:=\prod_{i=1}^n \mathcal{A}_i.
\]
$\mathcal{F}(\Theta)$ denotes the set of joint distributions on $\Theta$, while a multidimensional mechanism is a Borel probability kernel $M:\Theta\rightarrow\Delta(\mathcal{A})$. The principal knows neither the marginal state distributions nor the dependence between states across dimensions, and maximizes her worst-case payoff across all joint distributions. We write $\theta$ for a state vector and $\theta_i$ for its component in dimension $i$. For the ideal-action mappings $y_j$, we add a subscript $i$ to indicate the dimension and write $y_{j,i}$. We assume that utility is additively separable across dimensions:
\[
u_j(a,\theta)
=\sum_{i=1}^n u_{j,i}(a_i,\theta_i)
=\sum_{i=1}^n -(a_i-y_{j,i}(\theta_i))^2
\quad \text{for } j\in\{A,P\}.
\]
The expected utility from a state report $\hat{\theta}$ at state $\theta$ is denoted by $U_{j,i}^M(\theta,\hat{\theta})$ in dimension $i$ and $U_j^M(\theta,\hat{\theta})$ across all dimensions, so
\[
U_j^M(\theta,\hat{\theta})=\sum_{i=1}^n U_{j,i}^M(\theta,\hat{\theta}).
\]
Similarly, $V_i(M,F)$ denotes the principal's expected utility from mechanism $M$ and distribution $F$ in dimension $i$, so
\[
V(M,F)=\sum_{i=1}^n V_i(M,F).
\]
Multidimensional incentive compatibility requires
\[
U_A^M(\theta,\theta)\ge U_A^M(\theta,\hat{\theta})
\quad \forall \theta,\hat{\theta}\in\Theta.
\]
Let $\mathcal{M}$ denote the set of all incentive-compatible direct mechanisms.

We first assume that each player's ideal-action functions are identical across all dimensions, i.e., $y_{j,i}=y_{j,1}$ for all $i\in\{1,\ldots,n\}$. This is a natural assumption when the agent is responsible for several similar decisions. For example, an owner may use the same investment criterion for several comparable projects, while the manager favors the same degree of overinvestment in each project.

Let $(m_1^*,F_1^*)$ be a saddle point in the common one-dimensional problem, so
\begin{equation}\label{1d_saddle}
V_1(m_1^*,F_1)\ge V_1(m_1^*,F_1^*)\ge V_1(m_1,F_1^*)
\end{equation}
for all distributions $F_1\in\mathcal{F}(\Theta_1)$ and all one-dimensional incentive-compatible mechanisms $m_1$. Set $m_i^*(\hat{\theta}_i)=m_1^*(\hat{\theta}_i)$ for each dimension $i$. In particular, the mean and variance functions satisfy
\[
\mu_{m_i^*}(\hat{\theta}_i)=\mu_{m_1^*}(\hat{\theta}_i),
\qquad
\tau_{m_i^*}(\hat{\theta}_i)=\tau_{m_1^*}(\hat{\theta}_i).
\]
We construct the multidimensional candidate mechanism as the product of these one-dimensional mechanisms:
\[
M^*(\hat{\theta})=\bigotimes_{i=1}^n m_i^*(\hat{\theta}_i).
\]
We refer to this mechanism as \textit{separate delegation} because the lottery assigned in dimension $i$ depends only on the state report in dimension $i$.

\begin{proposition}
\label{prop:multidimensional-saddle}
If each player's ideal-action function is identical across dimensions, then the separate-delegation mechanism $M^*$ is robustly optimal.
\end{proposition}

Proposition~\ref{prop:multidimensional-saddle} shows that the mechanism $M^*$ achieves the highest payoff guarantee available to the principal, even when she can use mechanisms that condition each allocation on the entire vector of reports. Because the mechanism treats each dimension separately, the principal guarantees herself the sum of the one-dimensional saddle-point payoffs. This guarantee holds for every joint distribution of states, regardless of the dependence across dimensions. Separate delegation therefore provides a lower bound on the principal's optimal payoff in the multidimensional problem. Additionally, separate delegation takes a simple form: the principal applies the same one-dimensional mechanism to each decision, and the lottery assigned to a decision depends only on the report for that decision. 

To show that the principal cannot improve on this guarantee, we construct a worst-case distribution supported on the diagonal. Nature draws a state $t$ from the one-dimensional saddle-point distribution $F_1^*$ and sets $\theta_i=t$ in every dimension. The resulting joint distribution is
\[
F^*(\theta_1,\ldots,\theta_n)
=\int_{\Theta_1}\mathbb{I}_{\{t\leq\theta_1,\ldots,t\leq\theta_n\}}\,dF_1^*(t).
\]
Under this distribution, each player has the same ideal action across dimensions. The incentives generated by any incentive-compatible multidimensional mechanism can therefore be represented by a one-dimensional randomized mechanism with the same average payoff across dimensions. Such a mechanism cannot improve on the one-dimensional saddle-point payoff against $F_1^*$. No incentive-compatible mechanism can therefore give the principal more than the sum of the one-dimensional saddle-point payoffs against $F^*$. The proposed pair $(M^*,F^*)$ thus forms a saddle point.

Separate delegation remains optimal for a broader class of preferences. Ideal-action functions may differ across dimensions through a relabeling of states and a common shift and rescaling of both players' actions. The following assumption describes these transformations.

\begin{assumption}\label{assumption1}
For every $i\in\{1,\ldots,n\}$ and $j\in\{A,P\}$,
\[
y_{j,i}(\theta_i)=q_i\,y_{j,1}(\psi_i(\theta_i))+r_i,
\]
where $q_i,r_i\in\mathbb{R}$ and $\psi_i:\Theta_i\rightarrow\Theta_1$ is a $C^1$-diffeomorphism. Further, $q_i\psi_i^\prime(\theta_i)>0$ for all $\theta_i\in\Theta_i$ and $\mathcal{A}_i=q_i\mathcal{A}_1+r_i$. We normalize $(q_1,r_1)=(1,0)$ and take $\psi_1$ to be the identity mapping.
\end{assumption}

As a diffeomorphism, $\psi_i$ is differentiable and monotone. Intuitively, $\psi_i(\theta_i)$ relabels the state space, while $a\mapsto q_i a+r_i$ shifts and rescales actions. The bias therefore satisfies
\[
b_i(\theta_i)=q_i b_1(\psi_i(\theta_i)).
\]
For example, reflect the state space by setting
\[
\psi_i(\theta_i)=1-\theta_i,
\qquad
y_{j,i}(\theta_i)=q_i y_{j,1}(1-\theta_i)+r_i,
\qquad q_i<0.
\]
If, in addition, $r_i=1$, $q_i=-1$, and $y_{P,1}(\theta_1)=\theta_1$, then $y_{P,i}(\theta_i)=\theta_i$ and
\[
y_{A,i}(\theta_i)=1-y_{A,1}(1-\theta_i).
\]
Thus, $y_{A,i}$ is obtained by jointly reflecting the state and action spaces. For example, one could have $y_{A,1}(\theta_1)=\theta_1+\lambda$ and $y_{A,2}(\theta_2)=\theta_2-\lambda$, or $y_{A,1}(\theta_1)=2\theta_1$ and $y_{A,2}(\theta_2)=2\theta_2-1$. The extension therefore allows the agent's bias to have opposite signs across dimensions.\footnote{Fixing $y_{P,i}(\theta_i)=\theta_i$ for all $i\in\{1,\ldots,n\}$ illustrates the restrictions imposed by the assumption: only this joint reflection of state and action spaces ($\psi_i(\theta_i)=1-\theta_i$, $q_i=-1$, $r_i=1$) or the identity transformation ($\psi_i(\theta_i)=\theta_i$, $q_i=1$, $r_i=0$) is feasible. In particular, if the bias is constant, so that $y_{A,i}(\theta_i)=\theta_i+\lambda_i$, its absolute value must be identical across dimensions. Differentiating the bias gives
\[
b_i^\prime(\theta_i)=q_i\psi_i^\prime(\theta_i)b_1^\prime(\psi_i(\theta_i)).
\]
Since $q_i\psi_i^\prime(\theta_i)>0$, the sign of $b_i^\prime(\theta_i)$ is determined by the sign of $b_1^\prime(\psi_i(\theta_i))$. As $b_1$ is monotone by assumption, the transformation preserves the direction of monotonicity of the bias.}

Under Assumption~\ref{assumption1}, the utility functions can be rewritten as
\begin{equation}\label{prop1_util_rewritten}
u_{j,i}(a_i,\theta_i)
=-q_i^2\left(\frac{a_i-r_i}{q_i}-y_{j,1}(\psi_i(\theta_i))\right)^2.
\end{equation}
Let $(m_1^*,F_1^*)$ be a saddle point in the reference dimension. For each dimension $i$, define
\[
m_i^*(\hat{\theta}_i):=q_i\,m_1^*(\psi_i(\hat{\theta}_i))+r_i,
\]
where the right-hand side denotes the affine transformation of the lottery assigned by $m_1^*$. The mean and variance functions are
\[
\mu_{m_i^*}(\hat{\theta}_i)=q_i\mu_{m_1^*}(\psi_i(\hat{\theta}_i))+r_i,
\qquad
\tau_{m_i^*}(\hat{\theta}_i)=q_i^2\tau_{m_1^*}(\psi_i(\hat{\theta}_i)).
\]
As before, the candidate mechanism is the product of these one-dimensional mechanisms:
\[
M^*(\hat{\theta})=\bigotimes_{i=1}^n m_i^*(\hat{\theta}_i).
\]
We call the set of vectors $(\theta_1,\psi_2^{-1}(\theta_1),\ldots,\psi_n^{-1}(\theta_1))$, with $\theta_1\in\Theta_1$, the \textit{reparametrized diagonal}. Define the candidate distribution by
\[
F^*(\theta_1,\ldots,\theta_n)
=\int_{\Theta_1}\mathbb{I}_{\{t\leq\theta_1,\psi_2^{-1}(t)\leq\theta_2,\ldots,\psi_n^{-1}(t)\leq\theta_n\}}\,dF_1^*(t).
\]

The construction follows the same logic as in the identical case. The principal applies the appropriately transformed one-dimensional mechanism separately in each dimension. Nature draws a reference state according to $F_1^*$ and assigns the corresponding state to each dimension using the inverse transformations. After normalizing actions, each dimension presents the same one-dimensional incentive problem, with utility scaled by $q_i^2$, as shown in equation~\eqref{prop1_util_rewritten}. Separate delegation guarantees the sum of the one-dimensional saddle-point payoffs, and the proposed distribution prevents the principal from improving on this guarantee by linking dimensions. The pair $(M^*,F^*)$ therefore again forms a saddle point, establishing the robust optimality of separate delegation under Assumption~\ref{assumption1}.

If Assumption~\ref{assumption1} does not hold, Nature may no longer be able to align the incentive problems across dimensions in this way. A mechanism that links dimensions might then improve on the principal's guarantee from treating decisions separately.

\section{Conclusion}
We consider a delegation problem in which the principal does not know the distribution from which the underlying state is drawn. The principal can choose a general randomized mechanism and maximizes her worst-case expected payoff over all state distributions. Our main result characterizes a robustly optimal mechanism for a large class of quadratic preferences.
When the bias does not cross zero, the mechanism has up to three regions: accommodation, where the agent’s ideal action is taken; calibrated randomization, where each type receives a distinct lottery; and pooling, where all types receive the same lottery.
For a positive bias, in the calibrated randomization region, the mean allocation moves closer to the principal’s ideal action as the state increases. The resulting gain exactly offsets the additional variance required to preserve incentives, keeping her payoff constant throughout the region.

We then consider preferences that cross. We split the state space at the point of agreement and characterize a saddle point for each of the two resulting problems. With increasing bias, both mechanisms accommodate the agent near the crossing and can be joined directly. With decreasing bias, we add variance in the region with the smaller worst-case loss to preserve incentives across the crossing. The principal may therefore optimally assign positive variance even where her preferred action coincides with the agent’s.
We also extend our model to a multidimensional setting with additively separable preferences. When each player’s ideal-action functions are transformations of that player’s own reference function, with the same smooth relabeling of states and affine transformation of actions applied to both players in each dimension, we characterize a robustly optimal mechanism and show that separate delegation is optimal.
\appendix
\printbibliography
\section{Proofs of the Main Results}
\label{app:main-results}
\subsection{Proof of Lemma~\ref{lem:ic-characterization}}
\begin{proof}
For a mechanism $m$ with mean $\mu$ and variance $\tau$, write
\[
    U_A(\theta,\hat\theta)
    :=-\tau(\hat\theta)-[\mu(\hat\theta)-y_A(\theta)]^2,
    \qquad
    U_A(\theta)=U_A(\theta,\theta).
\]
Suppose first that $m$ is incentive compatible. For any
$\theta,\hat\theta\in\Theta$, the incentive constraints give
\begin{align*}
    -[\mu(\theta)-y_A(\theta)]^2-\tau(\theta)
    &\geq-[\mu(\hat\theta)-y_A(\theta)]^2-\tau(\hat\theta),\\
    -[\mu(\hat\theta)-y_A(\hat\theta)]^2-\tau(\hat\theta)
    &\geq-[\mu(\theta)-y_A(\hat\theta)]^2-\tau(\theta).
\end{align*}
Adding the two inequalities yields
\[
    [\mu(\theta)-\mu(\hat\theta)]
    [y_A(\theta)-y_A(\hat\theta)]\geq0.
\]
Since $y_A$ is strictly increasing, this condition implies that the mechanism's mean, $\mu$, is weakly increasing, giving $(\mathrm{IC}_1)$.

The derivative of the agent's utility at state $\theta$ when reporting $\hat\theta$, with respect to $\theta$, equals
\[
    \frac{\partial U_A(\theta,\hat\theta)}{\partial\theta}
    =2[\mu(\hat\theta)-y_A(\theta)]y_A'(\theta).
\]
This derivative is uniformly bounded, since $\mu$ is bounded and $y_A$ is continuously differentiable on $\Theta$. Moreover, utility is continuously differentiable, hence absolutely continuous, in the
state for every fixed report. Under incentive compatibility, truthful reporting maximizes $U_A(\theta,\cdot)$, so the envelope theorem of
\textcite{milgrom2002envelope} applies:
\begin{equation}
    U_A(\theta)
    =U_A(0)+\int_0^\theta
    \left.
    \frac{\partial U_A(s,\hat\theta)}{\partial s}
    \right|_{\hat\theta=s}\,ds
    \qquad\text{for every }\theta\in\Theta.
    \label{eq:app-ic-envelope}
\end{equation}
Condition $(\mathrm{IC}_2)$ follows by substituting the derivative and subtracting the identity at $\theta_1$ from the identity at $\theta_2$. By feasibility, the variance
\[
    \tau(\hat\theta)
    =\int_{\mathcal A}[a-\mu(\hat\theta)]^2\,dm(a\mid\hat\theta)
\]
must be nonnegative, giving $(\mathrm{VAR})$.

For sufficiency, suppose $(\mathrm{IC}_1)$, $(\mathrm{IC}_2)$, and $(\mathrm{VAR})$ hold. Condition $(\mathrm{IC}_2)$ gives the difference between truthful utilities at $\theta$ and $\hat\theta$. Combining this identity with the utility decomposition and simplifying
yields
\[
    U_A(\theta)-U_A(\theta,\hat\theta)
    =2\int_{\hat\theta}^{\theta}
    [\mu(s)-\mu(\hat\theta)]y_A'(s)\,ds
    \qquad\text{for every }\theta,\hat\theta\in\Theta.
\]
The right-hand side is weakly positive by $(\mathrm{IC}_1)$ and $y_A'\geq0$. If $\theta\geq\hat\theta$, the integrand is
nonnegative; if $\theta<\hat\theta$, the integrand is nonpositive and the limits of integration are reversed. Hence, $m$ is incentive compatible.
\end{proof}
\subsection{Proof of Proposition~\ref{prop:complete-pooling}}
\begin{proof}
The mechanism $m^*$ is incentive compatible because it assigns the same action after every report. Since $y_P(0)\leq y_P(\theta)\leq y_P(1)$, the principal's statewise utility is at least $-[y_P(1)-y_P(0)]^2/4$, with equality at $0$ and $1$. Thus, for $F^*$ assigning probability $1/2$ to each endpoint,
\[
    V(m^*,F)\geq V(m^*,F^*)
    =-\frac{[y_P(1)-y_P(0)]^2}{4}
    \qquad\text{for every }F\in\mathcal F(\Theta).
\]

It remains to show that no incentive-compatible  mechanism gives the principal a higher expected utility under $F^*$. Take any $m\in\mathcal M$, with mean $\mu$ and variance $\tau$. The agent's incentive constraint at state $0$ against report $1$ gives
\[
    -\tau(0)-[\mu(0)-y_A(0)]^2
    \geq-\tau(1)-[\mu(1)-y_A(0)]^2.
\]
Expanding the squares and rearranging yields
\[
    \tau(1)+\mu(1)^2-\tau(0)-\mu(0)^2
    \geq2y_A(0)[\mu(1)-\mu(0)].
\]
By Lemma~\ref{lem:ic-characterization}, $\mu(1)\geq\mu(0)$.
Using the preceding inequality and $y_A(0)\geq y_P(1)$, we can compare the principal's utility at state $1$ from the lotteries assigned to reports $0$ and $1$:
\begin{align*}
    U_P(1,0)-U_P(1)
    &=\tau(1)+\mu(1)^2-\tau(0)-\mu(0)^2
      -2y_P(1)[\mu(1)-\mu(0)]\\
    &\geq2[y_A(0)-y_P(1)][\mu(1)-\mu(0)]\\
    &\geq0.
\end{align*}

The principal's expected utility under $F^*$ is therefore bounded above by the payoff from using the report-$0$ lottery at both endpoints. Completing the square gives
\begin{align*}
    V(m,F^*)
    &=\tfrac12 U_P(0)+\tfrac12 U_P(1)\\
    &\leq\tfrac12 U_P(0)+\tfrac12 U_P(1,0)\\
    &=-\tau(0)-\frac12\left[
      [\mu(0)-y_P(0)]^2+[\mu(0)-y_P(1)]^2\right]\\
    &=-\frac{[y_P(1)-y_P(0)]^2}{4}-\tau(0)
      -\left[\mu(0)-\frac{y_P(0)+y_P(1)}2\right]^2\\
    &\leq-\frac{[y_P(1)-y_P(0)]^2}{4}=V(m^*,F^*).
\end{align*}
The last inequality follows from $\tau(0)\geq0$: neither a positive variance nor a mean different from the midpoint can improve this bound. Together with Nature's best response established above, this shows that $(m^*,F^*)$ is a saddle point.
\end{proof}

\subsection{Lemma \ref{lem:forward_bias_lemma} and Its Proof}

For $m\in\mathcal M$ and $F\in\mathcal F(\Theta)$, let
\[
    \theta_\diamond:=\min\operatorname{supp}(F),
    \qquad
    \theta^\diamond:=\max\operatorname{supp}(F).
\]
Define the forward bias induced by $F$ as
\[
    S_F(\theta)
    :=\int_{(\theta,\theta^\diamond]}
    [y_A(\theta)-y_P(t)]\,dF(t).
\]
Next, introduce the left-continuous versions of the mean and variance above $\theta_\diamond$:
\[
    (\mu_m^L(\theta),\tau_m^L(\theta))
    :=
    \begin{cases}
        (\mu_m(\theta),\tau_m(\theta)),
        & \theta\in[0,\theta_\diamond],\\
        (\mu_m(\theta^-),\tau_m(\theta^-)),
        & \theta\in(\theta_\diamond,1],
    \end{cases}
\]
where $\theta^-$ denotes the left limit. In particular, we have then
\[ d\mu_m^L([\theta_1,\theta_2))=\mu_m^L(\theta_2)-\mu_m^L(\theta_1), \quad  \theta_\diamond \leq\theta_1 \leq \theta_2\leq1 . \]
Let
\[
    \Gamma
    :=\left\{
        \theta\in(\theta_\diamond,\theta^\diamond]:
        \mu_m(\theta)>\mu_m(\theta^-)
        \text{ and } F(\theta)>F(\theta^-)
    \right\}
\]
be the set of common left-discontinuities of $\mu_m$ and $F$. This set is countable since both functions are nondecreasing.

\begin{lemma}
\label{lem:forward_bias_lemma}
We can write the principal's utility from the pair $(m,F)$ in the following form:
\begin{align*}
    V(m,F)= & -(\mathbb{E}_F[y_P(\theta)]-\mu_m(\theta_\diamond))^2 - \tau_m(\theta_\diamond)-\text{Var}_F[y_P(\theta)] -2\int_{[\theta_\diamond,\theta^\diamond]} S_F(\theta) d\mu_m^L(\theta) \\ & - 2 \sum_{\theta \in \Gamma } b(\theta) (F(\theta)-F(\theta^-))(\mu_m(\theta)-\mu_m(\theta^-))
\end{align*}    
\end{lemma}

\begin{proof}  
The Lebesgue-Stieltjes integral $\int_{[\theta_\diamond,\theta^\diamond]} S_F(\theta)d\mu_m^L(\theta)$ is well-defined and finite since $S_F(\theta)$ is clearly Borel-measurable as well as bounded and $\mu_m^L$ is monotone by IC from Lemma \ref{lem:ic-characterization}. Since the agent's truthful statewise utility stays the same, the difference between the principal's statewise utility under mechanism $m$ and its left-continuous transformation $m^L$ equals: 
\[ U_{P,m}-U_{P,m^L}=-2b(\theta)\left( \mu_m(\theta)-\mu_m^L(\theta) \right) \]

The principal's utility from a mechanism $m$ can be decomposed as follows utilizing equation \eqref{util_decomp} (for notational convenience, we will drop the $m$ subscript for $\mu, \tau$ within this proof): 
\begin{equation} \label{l1_expand}
\begin{aligned}
    V(m,F)  = & - \int_{[\theta_\diamond,\theta^\diamond]} (\mu(\theta)-y_P(\theta))^2 +\tau(\theta) \, dF(\theta)  \\
      = & - \int_{[\theta_\diamond,\theta^\diamond]} y_P(\theta)^2 \, dF(\theta) + \int_{[\theta_\diamond,\theta^\diamond]} 2\mu^L(\theta)y_P(\theta) \, dF(\theta) \\ & -\int_{[\theta_\diamond,\theta^\diamond]} \mu^L(\theta)^2 +\tau^L(\theta) \, dF(\theta) 
     - 2 \sum_{\theta \in \Gamma} b(\theta)(F(\theta)-F(\theta^-))(\mu(\theta)-\mu(\theta^-))
     \end{aligned} \end{equation}
We will first apply integration by parts for Lebesgue-Stieltjes integrals to the second term in the middle line; specifically, we apply it to  the terms $\mu^L(\theta)$ and $-\int_{(\theta,\theta^\diamond]} 2y_P(t) \, dF(t)$. Since $\mu^L(\theta)$ is left-continuous, this yields: 
\begin{equation*}
   \int_{[\theta_\diamond,\theta^\diamond]} 2\mu^L(\theta)y_P(\theta) \, dF(\theta)=2\mu(\theta_\diamond) \mathbb{E}_F[y_P(\theta)] + \int_{[\theta_\diamond,\theta^\diamond]} \int_{(\theta,\theta^\diamond]} 2y_P(t) \, dF(t) \, d\mu^L(\theta)
\end{equation*}
From the $(IC_2)$-condition of Lemma \ref{lem:ic-characterization}, we know that:
\begin{equation*}
    d U_{A,m}(\theta) = 2 y_A^\prime(\theta)(\mu^L(\theta)-y_A(\theta)) \, d\theta
\end{equation*}
Next, we differentiate the agent's utility from a mechanism $m$ from truth-telling, which gives us: 
\begin{equation*}
    d U_{A,m}(\theta,\theta)= -d\left( (\mu^L(\theta))^2 + \tau^L(\theta) \right) +2\mu^L(\theta)y_A^\prime(\theta) \, d\theta +2 y_A(\theta) \, d\mu^L(\theta) - 2y_A(\theta)y_A^\prime(\theta) \, d\theta
\end{equation*}
Equating the two equations yields:
\begin{equation}\label{l1_IC}
    d\left((\mu^L(\theta))^2+\tau^L(\theta)\right)=2y_A(\theta) \, d \mu^L(\theta)
\end{equation}
Now, we go back to the third term in equation \eqref{l1_expand}, apply integration-by-parts to $\mu^L(\theta)^2+\tau^L(\theta)$ and $-(1-F(\theta))$, and rewrite the integral based on equation \eqref{l1_IC}: 
\begin{equation*}
   \int_{[\theta_\diamond,\theta^\diamond]} \mu^L(\theta)^2 +\tau^L(\theta) \, dF(\theta) = \mu(\theta_\diamond)^2 + \tau(\theta_\diamond) + \int_{[\theta_\diamond,\theta^\diamond]} (1-F(\theta)) 2y_A(\theta) \, d\mu^L(\theta)
\end{equation*}
Combining these identities and completing the square gives: 
\begin{equation*} \begin{aligned}
    V(m,F)= & -(\mathbb{E}_F[y_P(\theta)]-\mu(\theta_\diamond))^2 -\tau(\theta_\diamond) - \text{Var}_F[y_P(\theta)] - 2\int_{[\theta_\diamond,\theta^\diamond]} S_F(\theta) \, d \mu^L(\theta) - \\ & 2 \sum_{\theta \in \Gamma} b(\theta)(F(\theta)-F(\theta^-))(\mu(\theta)-\mu(\theta^-))
    \end{aligned}
\end{equation*}
\end{proof}  

\subsection{Proof of Theorem \ref{thm:positive-bias-saddle}}
\begin{proof}
We will start by showing that $m^*$ is incentive compatible. The candidate variance on $[\underline\theta,\bar\theta]$ is
\[
    \tau^*(\theta)=d(\underline\theta)^2-d(\theta)^2,
\]
with zero variance below $\underline\theta$ and constant variance $\tau^*(\bar\theta)$ above $\bar\theta$.

First, we will show that the induced mean and variance functions are continuous. If $\underline\theta>0$, continuity and the definition of $\underline\theta$ as an infimum imply $d(\underline\theta)=b(\underline\theta)$. Hence
$\mu^*(\underline\theta)=y_A(\underline\theta)$ and
$\tau^*(\underline\theta)=0$, so both functions match at the lower cutoff. If the pooling region is nonempty, the definitions of
$d$ and $\bar\theta$ give
\[
    \mu^*(\bar\theta)
    =y_P(\bar\theta)+\frac{b(\bar\theta)}2
    =\frac{y_P(\bar\theta)+y_P(1)}2.
\]
The variance also matches by construction. If $b(1)=0$, then $\bar\theta=1$, and the integral representation of $d$ gives $0\leq d(\theta)\leq y_P(1)-y_P(\theta)\to0=d(1)$. Thus both functions are continuous at this endpoint as well.

To show that $(\mathrm{IC}_1)$ holds, we only need to show that $\mu^*$ is weakly increasing on $[\underline\theta,\bar\theta]$, as the remaining regions follow from its definition. The integral
representation of $d(\theta)$ gives $d(\theta)>0$ for
$0<\theta<\bar\theta$, since $y_P$ is strictly increasing. Differentiating that representation wherever $b(\theta)>0$ gives
\[
    d'(\theta)=\frac{y_P'(\theta)}{b(\theta)}
        [d(\theta)-b(\theta)],
    \qquad
    (\mu^*)'(\theta)
        =\frac{y_P'(\theta)}{b(\theta)}d(\theta)\geq0.
\]
Together with continuity at the cutoffs, this proves
$(\mathrm{IC}_1)$.

To show that $(\mathrm{VAR})$ holds, we only need to show that $d(\theta)$ is nonincreasing on $[\underline\theta,\bar\theta]$. By the preceding differential equation, this reduces to showing
$d(\theta)\leq b(\theta)$. If the bias is nondecreasing, then
\[
    (d-b)'(\theta)
    -\frac{y_P'(\theta)}{b(\theta)}[d(\theta)-b(\theta)]
    =-b'(\theta)\leq0.
\]
Multiplying by a positive integrating factor shows that $d\leq b$ continues to hold after any state in the set defining $\underline\theta$. The infimum definition therefore gives $d\leq b$ throughout the calibrated region. If the bias is nonincreasing, $d(\theta)>b(\theta)$ would imply $d'\geq0$
while $b'\leq0$. The difference would then remain strictly positive up to $\bar\theta$, contradicting
$d(\bar\theta)=b(\bar\theta)/2$. Thus $d'\leq0$ wherever $b>0$ on the calibrated region. If $b(0)=0$, this also implies $\underline\theta>0$; otherwise
$0<d(\bar\theta)\leq d(\theta)\leq b(\theta)\to0$ as
$\theta\downarrow0$ would be a contradiction. Hence
$b(\underline\theta)>0$, and continuity covers the possible zero-bias upper endpoint. Consequently, the calibrated variance is nonnegative, as are the zero accommodation variance and the constant pooling variance $\tau^*(\bar\theta)$.

Lastly, we need to prove that $(\mathrm{IC}_2)$ holds. The bounds $0\leq d\leq b$ give $0\leq(\mu^*)'\leq y_P'$ on the calibrated region. Together with continuity and the variance formula, this implies that both allocation functions are absolutely continuous. We can therefore go back to the differentiated envelope condition:
\[
    \frac{d}{d\theta}
    \left[-[\mu^*(\theta)-y_A(\theta)]^2-\tau^*(\theta)\right]
    =2[\mu^*(\theta)-y_A(\theta)]y_A'(\theta).
\]
Within $[\underline\theta,\bar\theta]$, this holds by construction from the variance formula and the differential equation for $d(\theta)$. Within $[0,\underline\theta)$, we have $(\tau^*)'(\theta)=0$
and $\mu^*(\theta)=y_A(\theta)$. In $(\bar\theta,1]$, both $\mu^*$ and $\tau^*$ are constant, reducing the equation to the trivial equality
$2[\mu^*(\theta)-y_A(\theta)]y_A'(\theta)
=2[\mu^*(\theta)-y_A(\theta)]y_A'(\theta)$.
Since $\mu^*$ and $\tau^*$ are continuous at the cutoffs $\underline\theta$ and $\bar\theta$, $(\mathrm{IC}_2)$ holds on the entire state space. Hence, the mechanism $m^*$ is IC.

Now we continue by proving that $V(m^*,F^*)\leq V(m^*,F)$ for all $F\in\mathcal F(\Theta)$. The preceding bounds give $0<d(\underline\theta)/b(\underline\theta)\leq1$, so the proposed $F^*$ is a CDF supported on
$[\underline\theta,\bar\theta]\cup\{1\}$. For Nature's best response, it suffices to show that
\[
    U_{P,m^*}(\theta')\leq U_{P,m^*}(\theta)
    \qquad\text{for all }(\theta',\theta)
    \in\operatorname{supp}(F^*)\times\Theta.
\]
Note that for all $\theta'\in[\underline\theta,\bar\theta]$,
we have
\[
    U_{P,m^*}(\theta')
    =-d(\theta')^2-d(\underline\theta)^2+d(\theta')^2
    =-d(\underline\theta)^2,
\]
which is clearly independent of $\theta'$. Further,
$\mu^*(\bar\theta)=[y_P(\bar\theta)+y_P(1)]/2$ and the
constant variance on $[\bar\theta,1]$ imply
\[
    U_{P,m^*}(\theta')=U_{P,m^*}(1)
    \qquad\text{for all }\theta'\in[\underline\theta,\bar\theta].
\]
For all $\theta\in(\bar\theta,1)$, utility is weakly higher: the mean is the midpoint of the two endpoint ideal actions, while $y_P(\theta)$ lies between them.

If $\underline\theta=0$, we are done. What remains to be shown is that utility on $[0,\underline\theta)$ is weakly higher than on $[\underline\theta,\bar\theta]$. By the cutoff definition and the preceding comparison when the bias is nonincreasing, $\underline\theta>0$ implies that the bias is nondecreasing. Further, we have $\mu^*(\underline\theta)=y_A(\underline\theta)$
and $\tau^*(\underline\theta)=0$. Taken together, this yields
\[
    U_{P,m^*}(\underline\theta)=-b(\underline\theta)^2
    \leq-b(\theta)^2=U_{P,m^*}(\theta)
    \qquad\text{for all }\theta\in[0,\underline\theta),
\]
which completes the proof of $V(m^*,F^*)\leq V(m^*,F)$.
In particular, $V(m^*,F^*)=-d(\underline\theta)^2$.

Secondly, we will show that $V(m^*,F^*)\geq V(m,F^*)$ for all $m\in\mathcal M$. For this, we use Lemma \ref{lem:forward_bias_lemma}. First, we will show that the forward bias $S_{F^*}(\theta)$ is zero on $[\underline\theta,\bar\theta]$ and nonnegative
on $[\underline\theta,1]$. Clearly, it is nonnegative for $\theta\in(\bar\theta,1)$, since $y_A(\theta)>y_P(1)$; at $1$, the integral is zero.

We will first show that the forward bias is constant on
$[\underline\theta,\bar\theta]$. On the interior, differentiating the proposed CDF gives
\[
    (F^*)'(\theta)
    =\frac{y_A'(\theta)}{b(\theta)}[1-F^*(\theta)].
\]
Differentiating the forward bias hence gives
\[
    S_{F^*}'(\theta)
    =y_A'(\theta)[1-F^*(\theta)]
      -b(\theta)(F^*)'(\theta)=0.
\]
By continuity from within the region, this yields a constant forward bias on $[\underline\theta,\bar\theta]$. Next,
$S_{F^*}(\bar\theta)=0$ since $y_A(\bar\theta)=y_P(1)$ and $F^*$ has no probability on $(\bar\theta,1)$. Thus the forward bias is zero throughout the calibrated region, also when
$\bar\theta=1$.

Using the notation of Lemma \ref{lem:forward_bias_lemma}, for any $m\in\mathcal M$ we have
\begin{align*}
    V(m,F^*)
    &=-\bigl(\mathbb E_{F^*}[y_P(\theta)]
                 -\mu_m(\underline\theta)\bigr)^2
          -\tau_m(\underline\theta)
          -\operatorname{Var}_{F^*}[y_P(\theta)]\\
    &\quad-2\int_{[\underline\theta,1]}
          S_{F^*}(\theta)\,d\mu_m^L(\theta)\\
    &\quad-2\sum_{\theta\in\Gamma}
          b(\theta)[F^*(\theta)-F^*(\theta^-)]
          [\mu_m(\theta)-\mu_m(\theta^-)]\\
    &\leq-\operatorname{Var}_{F^*}[y_P(\theta)].
\end{align*}
The integral contribution is nonpositive since $S_{F^*}\geq0$ and $\mu_m^L$ is nondecreasing; every jump contribution is nonpositive since the bias and the jumps of both $F^*$ and $\mu_m$ are nonnegative.
To show that $m^*$ attains this bound, note that zero forward bias at $\underline\theta$ and the definition of $F^*$ imply
\begin{align*}
    \mathbb E_{F^*}[y_P(\theta)]
    &=F^*(\underline\theta)y_P(\underline\theta)
      +[1-F^*(\underline\theta)]y_A(\underline\theta)\\
    &=y_P(\underline\theta)+d(\underline\theta)
      =\mu^*(\underline\theta).
\end{align*}
Hence $m^*$ attains the bound, as $\mu^*$ is continuous and constant on $[\bar\theta,1]$, $\tau^*(\underline\theta)=0$, and $\mu^*(\underline\theta)=\mathbb E_{F^*}[y_P(\theta)]$.
This proves the second saddle-point inequality.
\end{proof}

\subsection{Proof of Theorem \ref{thm:crossing-bias-saddle}}
\begin{proof}
As before, we prove the mechanism $\mu^*$ is robustly optimal via the following three steps. First, $\mu^*$ is incentive compatible. Second, $\mu^*$ guarantees the principal's loss is no larger than $\bar{w}$. Last, $F^*$ guarantees the principal's loss is no lower than $\bar{w}$. 
\paragraph{Incentive compatibility of $\mu^*$}
We first verify $\mu^*$ is nondecreasing. $\mu^*_1$ and $\mu_2^*$ are weakly increasing, which follows from their incentive-compatibility. Thus, monotonicity of $\mu^*$ boils down to continuity of $\mu^*$ at $\theta_{\times}$. If the bias is increasing, we prove continuity by showing that the accommodating region in the region with positive bias is nonempty and the nonemptiness of the accommodating region for the region with a negative bias follows by reflection. 

We characterize the regional saddle point for the positive region accordingly. As before $\bar{\theta} = y_A^{-1}(y_P(1))$, but now 
\[
\underline{\theta}
:=
\inf\left\{
\theta\in(\theta_\times,\bar{\theta}]:
d(\theta)\le b(\theta)
\right\}.
\]
We will prove $\underline{\theta}>\theta_{\times}$. First, we show whenever $d(\theta)-b(\theta)\leq 0$, then $d(\theta')-b(\theta')\leq0$ for all $\theta'\in(\theta,1]$. We will need this observation to establish the contradiction. We have established $(d-b)' = \frac{y_P'}{b}(d-b)-b'$. Thus, with increasing bias, whenever $d\leq b$, $(d-b)$ is decreasing. Moreover, $d$ is decreasing whenever $d\leq b$ since 
\[
d'(\theta)=y_P'(\theta)
\left(\frac{d(\theta)}{b(\theta)}-1\right)\le0.
\]
Now suppose $\underline{\theta} = \theta_{\times}$. By the definition of $\underline{\theta}$, there exist points $\theta_n\downarrow\underline{\theta}$ such that $d(\theta_n)\le b(\theta_n)$. The preceding argument suggests
\[
0<
\frac{b(\bar{\theta})}{2}
=d(\bar{\theta})
\le d(\theta_n)
\le b(\theta_n)
\to b(\underline{\theta})=0,
\]
The second inequality follows from the monotonicity of $d$ over $[\underline{\theta},\bar{\theta}]$. This is a contradiction. We have established that the accommodating region on the positive bias side is nonempty. Therefore, 
\[
\mu^*_1(\theta_\times) = \mu^*_2(\theta_{\times})=y_A(\theta_{\times})=y_P(\theta_{\times})
\]
The monotonicity of $\mu^*$ is established for increasing bias.

Now we prove the monotonicity of $\mu^*$ when the bias is decreasing. Similarly, the monotonicity of $\mu^*$ boils down to the continuity of $\mu^*$ at $\theta_\times$. We prove $\mu^*_1(\theta_\times) = y_A(\theta_\times)$; the proof for $\mu^*_2(\theta_\times) = y_A(\theta_\times)$ follows from the reflection argument. Now within the positive bias region $[0,\theta_\times]$, the pooling cut-off $\bar{\theta}$ and distance function are defined as 
\[
\bar{\theta}:=\theta_\times,\quad d(\theta)=
\int_\theta^{\theta_\times}
y_P'(t)\exp\!\left(
-\int_\theta^t\frac{y_P'(s)}{b(s)}\,ds
\right)\,dt.
\]
Thus, $\mu^*_1(\theta_\times)=y_A(\theta_\times)$ since $\lim_{\theta\to\theta_\times}d(\theta)=0$. 

Now we prove $IC_2$. Still, we first consider the case of increasing bias, $b'\geq 0$. To show $IC_2$, it suffices to consider $\theta_1,\theta_2\in\Theta$ with
$\theta_1\leq\theta_\times\leq\theta_2$, since $(IC_2)$ holds for pairs in the same region by incentive compatibility of the regional mechanisms $m_1^*$ and $m_2^*$. We must show that
\begin{equation}\label{p2_env}
U_{A,m^*}(\theta_2)-U_{A,m^*}(\theta_1)
=
2\int_{\theta_1}^{\theta_2}
[\mu^*(s)-y_A(s)]\,y_A'(s)\,ds.
\end{equation}
Since the regional mechanisms satisfy $(IC_2)$, we have
\begin{align*}
U_{A,m_2^*}(\theta_2)-U_{A,m_2^*}(\theta_\times)
&=
2\int_{\theta_\times}^{\theta_2}
[\mu_2^*(s)-y_A(s)]\,y_A'(s)\,ds, \text{ and } \\
U_{A,m_1^*}(\theta_\times)-U_{A,m_1^*}(\theta_1)
&=
2\int_{\theta_1}^{\theta_\times}
[\mu_1^*(s)-y_A(s)]\,y_A'(s)\,ds.
\end{align*}
Adding these equalities and using the construction of $\mu_{m^*}$, we can replace $\mu_i^*$ with $\mu^*$ in the corresponding regions:
\begin{equation}\label{p2_env_confirm}
U_{A,m_2^*}(\theta_2)-U_{A,m_2^*}(\theta_\times)
+U_{A,m_1^*}(\theta_\times)-U_{A,m_1^*}(\theta_1) =
2\int_{\theta_1}^{\theta_2}
[\mu^*(s)-y_A(s)]\,y_A'(s)\,ds.
\end{equation}
Moreover, by the regional construction in Theorem \ref{thm:positive-bias-saddle} and its reflection, both regional mechanisms accommodate the agent at $\theta_\times$. Hence,
\[
\mu_1^*(\theta_\times)
=
\mu_2^*(\theta_\times)
=
y_A(\theta_\times)
=
y_P(\theta_\times), \text{ and } 
\tau_1^*(\theta_\times)
=
\tau_2^*(\theta_\times)
=
0.\]
Thus, $U_{A,m_1^*}(\theta_\times)=U_{A,m_2^*}(\theta_\times)=0$. Together with the construction of $m^*$, this also implies $
\mu^*(\theta)=\mu_i^*(\theta)$ and $
\tau^*(\theta)=\tau_i^*(\theta)$ for all $\theta\in\Theta_i,i=1,2$.
Consequently,
\[
U_{A,m_2^*}(\theta_2)=U_{A,m^*}(\theta_2)
\quad\text{and}\quad
U_{A,m_1^*}(\theta_1)=U_{A,m^*}(\theta_1).
\]
The crossing-point utility terms in \eqref{p2_env_confirm} therefore cancel, yielding
\eqref{p2_env}. Reversing the endpoints gives the same identity for the opposite ordering. Thus $(IC_2)$ holds globally.

Now we verify $IC_2$ for decreasing bias. Without loss of generality, assume that $w_2\geq w_1$. By the regional construction and its reflection, $\mu_{m_i^*}(\theta_\times)=y_A(\theta_\times)=y_P(\theta_\times)$, and $\tau_{m_i^*}(\theta_\times)=w_i, i=1,2$ as $\theta_\times$ lies in the calibrated randomization region on each side. Hence, $U_{A,m_i^*}(\theta_\times)=-w_i$. In particular, after the variance adjustment the two regional lotteries have an identical mean and variance at $\theta_\times$: $\tau_{m_1^*}(\theta_\times)+ w_2- w_1 =\tau_{m_2^*}(\theta_\times)$. The construction of $m^*$ gives
\[
U_{A,m^*}(\theta)=
\begin{cases}
U_{A,m_1^*}(\theta)-( w_2- w_1),
    & \theta\in\Theta_1,\\
U_{A,m_2^*}(\theta),
    & \theta\in\Theta_2.
\end{cases}
\]
Thus, $U_{A,m_1^*}(\theta_\times)-(w_2-w_1)=U_{A,m_2^*}(\theta_\times)=-w_2$. Since the adjustments are constant within each region, $(IC_2)$ continues to hold for pairs in the same region. Now take $\theta_1\leq\theta_\times\leq\theta_2$. Substituting the preceding utility identities into \eqref{p2_env_confirm} yields
\begin{align*}
&U_{A,m^*}(\theta_2)
-U_{A,m_2^*}(\theta_\times)
+U_{A,m_1^*}(\theta_\times)
-U_{A,m^*}(\theta_1)-(w_2-w_1)\\
&\qquad =
2\int_{\theta_1}^{\theta_2}
[\mu_{m^*}(s)-y_A(s)]\,y_A'(s)\,ds.
\end{align*}
Since
\[
-U_{A,m_2^*}(\theta_\times)
+U_{A,m_1^*}(\theta_\times)-(w_2-w_1)
=w_2-w_1-(w_2-w_1)=0,
\]
equation \eqref{p2_env} follows, verifying condition $IC_2$. $\tau^*(\theta)\ge0$ for all $\theta \in \Theta$, which follows from the fact that the variance $\tau^*_i$ is nonnegative for each $i=1,2$.
\paragraph{Nature's guarantee}
To prove Nature's guarantee, we need to show
\[
V(m,F^*)\leq V(m^*,F^*) = -\bar{w},\quad\forall m\in\mathcal{M}
\]
Assume without loss of generality that $\bar{w}=w_2$; then by definition $F^*=F_2^*$. Now, restrict an arbitrary competing mechanism to $\Theta_2$. If $m$ is globally incentive compatible, its restriction
\[m|_{\Theta_2}:\Theta_2 \rightarrow \Delta(\mathcal{A})\]
is incentive-compatible for the regional problem. By Theorem \ref{thm:positive-bias-saddle}, we have
\[
V(m,F^*)=V(m|_{\Theta_2},F_2^*)\leq V(m_2^*,F_2^*)=-\bar w=V(m^*,F^*).
\] 

\paragraph{The principal's guarantee}
To verify the principal's guarantee, we need to prove 
\[
V(m^*,F)\geq V(m^*,F^*),\quad\forall F\in\mathcal{F}(\Theta)
\]
By Theorem \ref{thm:positive-bias-saddle}, each regional mechanism guarantees utility at least $-w_i$ against every distribution supported on $\Theta_i$. In particular, applying this guarantee to the Dirac distribution at any $\theta\in\Theta_i$ gives
\[U_{P,m_i^*}(\theta)
=V(m_i^*,\delta_\theta)
\ge -w_i.\]
If the bias is increasing, the combined mechanism preserves the regional means and variances. Therefore, for every $\theta\in\Theta_i$,
\[
U_{P,m^*}(\theta)
=U_{P,m_i^*}(\theta)
\ge-w_i
\ge-\bar w.
\]
If the bias is decreasing, the combined mechanism preserves the regional mean and increases the variance by $\bar w-w_i$. Hence, for every $\theta\in\Theta_i$,
\[U_{P,m^*}(\theta)= U_{P,m_i^*}(\theta)-(\bar w-w_i)\ge-w_i-(\bar w-w_i)=-\bar w.\]
Thus, in either case,
\[U_{P,m^*}(\theta)\ge-\bar{w}
\ \text{ for every }\theta\in\Theta.\]
Integrating against any probability distribution $F$ yields
\[
V(m^*,F)
=\int_\Theta U_{P,m^*}(\theta)\,dF(\theta)
\ge-\bar w=V(m^*,F^*).
\]
\end{proof}

\subsection{Proof of Proposition~\ref{prop:multidimensional-saddle} under Assumption \ref{assumption1}}
\begin{proof}

Now, let us briefly establish that the pair $(m_i^*,F_i^*)$ for $i \in \{1,...,n\}$ forms a saddle-point in each dimension $i$. To see this, note that from equation \eqref{prop1_util_rewritten}, we obtain an expected payoff-equivalence identity: 
\begin{multline*}
    -(q_i)^2 \int_{\Theta_i} \int_{\mathcal{A}_i} \left( \frac{a_i-r_i}{q_i} - y_{A,1}(\psi_i(\theta_i)) \right)^2 \, dm_i^*(a_i|\theta_i) \, dF_i^*(\theta_i)  = \\
    -(q_i)^2 \int_{\Theta_1} \int_{\mathcal{A}_1} \left(a_1 - y_{A,1}(\theta_1) \right)^2 \, dm_1^*(a_1|\theta_1) \, dF_1^*(\theta_1).
\end{multline*} 
As $m_1^*$ is $IC$, $IC$ for $m_i^*$ follows immediately. Note that in dimension $i=1$, we have $q_1=1,r_1=0$, and $\psi_1$ is the identity mapping. 
By utilizing the same rewriting of utilities (which clearly works identically for the principal's utility and for any pair $(m_i,F_i)$), we can establish easily that $(m_i^*,F_i^*)$ being a saddle-point follows immediately from $(m_1^*,F_1^*)$ forming a saddle-point in the reference dimension. 

\noindent
Next, we will prove that our candidate $M^*$ is incentive-compatible: Let's denote the arbitrary state vector by $\theta \in \Theta$ and its report by $\hat{\theta} \in \Theta$. From equation \eqref{prop1_util_rewritten}:
\[ U_{A,i}^{m_i^*}(\theta_i,\hat{\theta}_i)=q_i^2 \; U_{A,1}^{m_1^*}(\psi_i(\theta_i),\psi_i(\hat\theta_i)) \]
$m_1^*$ being incentive-compatible in one dimension yields that for each pair $(\theta_i,\hat{\theta}_i)$,
\begin{equation}\label{IC_1D}
    U_{A,1}^{m_1^*}(\psi_i(\theta_i),\psi_i(\theta_i))
   \ge
U_{A,1}^{m_1^*}(\psi_i(\theta_i),\psi_i(\hat\theta_i))
\end{equation}
Since $M^*$ is by construction just the product of intra-dimensional lotteries $m_i^*$, we can sum equation \eqref{IC_1D} over all n dimensions and weight each dimension by $q_i^2$, which directly yields IC:
\begin{multline*}
     U_A^{M^*}(\theta,\theta) =\sum_{i=1}^n q_i^2 \, U_{A,1}^{m_1^*}(\psi_i(\theta_i),\psi_i(\theta_i)) \ge  \\ 
     \sum_{i=1}^n q_i^2 \, U_{A,1}^{m_1^*}(\psi_i(\theta_i),\psi_i(\hat\theta_i)) =U_A^{M^*}(\theta,\hat{\theta}) \quad  \forall \theta,\hat{\theta} \in \Theta
\end{multline*} 

\noindent
Next, we will show that $F^*$ is a best-response to $M^*$, so $V(M^*,F)\ge V(M^*,F^*)$ for all $F \in \mathcal{F}(\Theta)$. Let $F$ be an arbitrary joint CDF on $\Theta$, with marginals $F_1,...,F_n$. 
As the candidate $M^*$ maps dimensions separately, the action lottery in dimension $i$ depends solely on the respective state-report $\hat{\theta}_i$; we can infer: 
\[ V(M^*,F)
   =\sum_{i=1}^n V_{i}(m_i^*,F_i) \]
As we have established $(m_i^*,F_i^*)$ to be a saddle point in each dimension $i$, we get for all $F \in \mathcal{F}(\Theta)$: 
\[V(M^*,F)
   =\sum_{i=1}^n V_{i}(m_i^*,F_i) \ge \sum_{i=1}^n V_{i}(m_i^*,F_i^*)=V(M^*,F^*)  \]

\noindent
Next, we will show that $M^*$ is a best-response to $F^*$, so $V(M^*,F^*)\ge V(M,F^*)$ for all $M \in \mathcal{M}$.
Take any multidimensional IC mechanism
\[M:\Theta\to\Delta(\mathcal{A})\]
Write any point on the reparametrized diagonal as $\theta_{d}=(t, \psi_2^{-1}(t),...,\psi_n^{-1}(t))$ (for any $t \in \Theta_1$) and the agent's diagonal report given this diagonal realization as $\hat{\theta}_{d}=(\hat{t}, \psi_2^{-1}(\hat{t}),...,\psi_n^{-1}(\hat{t}))$ (for any $\hat{t} \in \Theta_1$). Further, let $\Theta_{d,\psi}$ denote the subset of the state space that only contains reparametrized diagonal arguments under the vector $\boldsymbol{\psi}$.   

For the mechanism $M$, denote its $i$-th marginal allocation lottery by $M_i(\hat{\theta})$ following the report $\hat{\theta}$.
Now construct the following auxiliary one-dimensional stochastic mechanism $\phi$ induced by mechanism $M$: 
\begin{equation*}
\phi_{M}(\hat{t}):= \sum_{i=1}^n \frac{q_i^2}{\sum_{k=1}^n q_k^2} \int_{\mathcal{A}_i} \delta_{\frac{a_i-r_i}{q_i}}  \,  dM_i(a_i|\hat{\theta}_d)    
\end{equation*}
This auxiliary mechanism $\phi_{M}(\hat{t})$ chooses one dimension $i \in \{1,..,n\}$ with probability \linebreak $ ({q_i^2})/({\sum_{k=1}^n q_k^2})$, takes the outcome under the original mechanism $M$ in this dimension $i$, and implements its rescaled/transformed version. Hence, it reflects a mixture over the rescaled stochastic allocation of the mechanism $M$ after the report $\hat{\theta}_d$.

Next, we will show that this auxiliary one-dimensional mechanism $\phi_{M}$ is one-dimensional IC. We start by applying multi-dimensional IC, which yields for every multi-dimensional mechanism $M \in \mathcal{M}$: 
\begin{multline*}
\sum_{i=1}^n q_i^2
 \int_{\mathcal{A}_i} u_{A,1}\left(\frac{a_i-r_i}{q_i},t\right)\, d M_i(a_i|\theta_d)
\ge \\
\sum_{i=1}^n q_i^2
 \int_{\mathcal{A}_i} u_{A,1}\left(\frac{a_i-r_i}{q_i},t\right)\, d M_i(a_i|\hat{\theta}_d) \quad  \forall \theta_d,\hat{\theta}_d \in \Theta_{d,\psi} 
\end{multline*}
Dividing both sides by ${\sum_{k=1}^n q_k^2}$ gives:
\[
U_{A,1}^{\phi_M}(t,t)
   \ge U_{A,1}^{\phi_M}(t,\hat{t})  \quad  \forall t,\hat{t} \in \Theta_{1}
\]
Therefore, $\phi_{M}$ is $IC$. 
Since the auxiliary mechanism $\phi_{M}$ just gives a mixture of the rescaled marginal stochastic allocations of the mechanism $M$ under distribution $F^*$, we have for every $M \in \mathcal{M}$: 
\begin{equation*}
V(M,F^*)
   = \sum_{i=1}^n q_i^2 \, V_{1}(\phi_{M},F_1^*)
\end{equation*}
Utilizing the second inequality from \eqref{1d_saddle}, we get: 
\begin{equation*}
V(M,F^*)
  = \sum_{i=1}^n q_i^2 \, V_{1}(\phi_{M},F_1^*)
  \le \sum_{i=1}^n q_i^2 \, V_{1}(m_1^*,F_1^*)
  =V(M^*,F^*)
\end{equation*}
Consequently, we have: 
\begin{equation*}
V(M^*,F)
\ge V(M^*,F^*)
\ge V(M,F^*) \quad
\forall M \in \mathcal{M}, F \in \mathcal{F}(\Theta)
\end{equation*}
Therefore,  $(M^*,F^*)$ is a saddle point. 
\end{proof}

\section{Corollary for the Negative Bias}
\label{app:negative-bias}
We now extend our result from the positive bias to negative bias. Assume that $b(\theta)\leq0$. We reflect both states and actions. Let 
\[R(\theta):=1-\theta, \qquad
\kappa:=y_P(0)+y_P(1), \qquad a_R:=\kappa-a\]
The reflected ideal actions and bias are given by
\begin{equation*}
    y_{i,R}(\theta)=\kappa-y_i(1-\theta),
    \qquad
    b_R(\theta)=-b(1-\theta)\geq0.
    \label{eq:reflected-preferences}
\end{equation*}

Hence, $y_A\leq y_P$ in the non-reflected environment is equivalent to $y_{A,R}\geq y_{P,R}$ in the reflected one. By setting $\kappa=y_P(0)+y_P(1)$, the reflection preserves the principal's ideal actions at the endpoints, so $y_{P,R}(0)=y_P(0)$ and $y_{P,R}(1)=y_P(1)$. Let $U^R$ and $V^R$ denote the reflections of the respective utility functions. 
Further, we have payoff equivalence between the original and the reflected delegation problems as for every pair $(\theta,\hat\theta)$: 
\begin{align*}
U_{j,m_R}^R(R(\theta),R(\hat\theta)) & = \int_{\mathcal{A}_R} -\left[ a_R-y_{j,R}(R(\theta)) \right]^2 \, dm_R(a_R|R(\hat{\theta})) \\ & = \int_\mathcal{A} -\left[a-y_j(\theta)\right]^2 \, dm(a|\hat{\theta}) =U_{j,m}(\theta,\hat\theta)
\end{align*}
A mechanism's reflection $m_R=(\mu_R,\tau_R)$ will be defined by
\[\mu_R(\theta):=\kappa-\mu(1-\theta),
\quad \quad
\tau_R(\theta):=\tau(1-\theta).
\]
For any CDF, let $F_R$ denote its reflection under $R(\theta)$; so, mass at $\theta$ is mapped to $1-\theta$; hence $F_R(\theta)=1-F^-(1-\theta)$. Because reflection is an involution, reflecting $m_R$ and $F_R$ again yields $m$ and $F$, respectively.

\begin{corollary}
    Suppose $y_A(\theta)\le y_P(\theta)$ for all $\theta\in\Theta$ and let $(m^*,F^*)$ be a saddle point constructed in Theorem~\ref{thm:positive-bias-saddle} for the reflected preferences $(y_{A,R},y_{P,R})$. Then $((m^*)_R,(F^*)_R)$ forms a saddle point for $(y_{A},y_{P})$.
\end{corollary}

Pointwise payoff equivalence implies that (i) $m$ is incentive-compatible if and only if $m_R$ is incentive-compatible and (ii) for any distribution $F$, $V(m,F)=V^R(m_R,F_R)$. Since the reflection $R$ is an involution, the two best-response inequalities for $(m^*,F^*)$ imply the corresponding ones for $((m^*)_R,(F^*)_R)$. \par

\end{document}